\documentclass[12pt]{amsart}
\usepackage{amsmath, amsthm, amsfonts, amsbsy, amssymb, upref, enumerate, bigstrut,  color, mathtools, mathrsfs, float, bm, dsfont, float}
\usepackage{psfrag}
\usepackage{natbib}

\usepackage[left=1.3in,top=1.3in,right=1.3in]{geometry}

\usepackage{epsfig, graphicx}
\usepackage{ hyperref}

\newtheorem{theorem}{Theorem}[section]

\newtheorem{corollary}[theorem]{Corollary}
\newtheorem{proposition}[theorem]{Proposition}

\newtheorem{remark}[theorem]{Remark}

\newcommand{\vn}{{\bf v}}
\newcommand{\taun}{{\mbox{\boldmath $\tau$}}}
\newcommand{\psin}{{\mbox{\boldmath $\psi$}}}
\newcommand{\tetn}{{\mbox{\boldmath $\theta$}}}

\begin{document}
\title[Bimodal Gamma Distribution]
{A bimodal gamma distribution: Properties, regression model and applications}
\author{R. Vila, L. Ferreira, H. Saulo, F. Prataviera and E.M.M. Ortega}
\address{
    \newline
    Departamento de Estat\'istica -
Universidade de Bras\'ilia, 70910-900, DF, Brazil,
    \newline
    {\rm Roberto \ Vila}, \quad Email: \textup{\tt rovig161@gmail.com}
    \newline
    {\rm Let\'icia \ Ferreira}, \quad  Email: \textup{\tt leticia.ferreira.reiss@gmail.com }
    \newline
    {\rm Helton Saulo}, \quad Email: \textup{\tt heltonsaulo@gmail.com }
 \newline
 \newline
    Departamento de Estat\'istica -
 Universidade de S\~ao Paulo,  05508-220, SP, Brazil,
   \newline
{\rm F\'abio \ Prataviera}, \quad Email: \textup{\tt fabio$\_$prataviera@usp.br}
\newline
{\rm Edwin \ Ortega}, \quad  Email: \textup{\tt edwin@usp.br }
}

\date{\today}


\begin{abstract}
In this paper we propose a bimodal gamma distribution using a quadratic transformation based on the alpha-skew-normal model. We discuss several properties of this distribution such as mean, variance, moments, hazard rate and entropy measures. Further, we propose
a new regression model with censored data based on the bimodal gamma
distribution. This regression model can be very useful to the analysis of real data and could give more realistic fits than other special regression models. Monte Carlo simulations were performed to check the bias in the maximum likelihood estimation. 
The proposed models are applied to two real data sets found in literature.
\end{abstract}

\maketitle
\section{Introduction}
\label{sec:1}

The unimodal gamma distribution is well known due to
its flexibility and good properties \citep{jkb:94}. This model has been widely applied
in several areas, such as physics \citep{Ismadji2000}, medicine
\citep{Shankar2003,bp:06}, quality control \citep{hpw:08,
deryac:12}, and inventory \citep{nch:99,ms:02}, among others.

A general and effective way to introduce bimodality into a unimodal
distribution is through a quadratic transformation, as it demands
less computational effort in parameter estimation when compared to
mixture-based bimodal models. In this sense, \cite{e:10} introduced
a prominent quadratic transformation in the normal distribution that
produces asymmetry and bimodality. This transformation gave rise to
the alpha-skew-normal (ASN) family of distributions. A random
variable $Z$ has an ASN distribution with parameter $\delta$, if its
probability density function (PDF) and cumulative distribution
function (CDF) are given, respectively, by
\begin{equation}\label{sec2:01}
g(z)=\frac{(1-\delta{z})^2+1}{2+\delta^2}\,\phi(z)\,\, \mbox{and} \,\,
G(z)=\Phi(z)+\delta\left(\frac{2-\delta{z}}{2+\delta^2}\right)\phi(z), \quad x,\delta\in\mathbb{R},
\end{equation}
where $\delta$ is an asymmetric parameter that controls the
uni-bimodality effect and $\phi(\cdot)$ and $\Phi(\cdot)$ are the standard normal PDF and CDF, respectively; see \cite{e:10}. We denote
$X\sim\text{ASN}(\delta)$.


In this context, we introduce a bimodal gamma
(BGamma) distribution through the multiplication of the gamma density by
a quadratic function proposed by \cite{e:10}.  We present a
statistical methodology based on the proposed BGamma model including
model formulation, mathematical properties and estimation based on the
maximum likelihood (ML) method. Numerical evaluation is carried out by
both Monte Carlo simulation and application to real data. In
special, the proposed BGamma model provides better adjustment compared
to the mixture generalized gamma distribution propose by
\cite{cancaya2015}.

Survival analysis is one of the areas of statistics that has grown steadily in recent
decades. It is common for the response variable (time until the occurrence of the event of
interest) to be related to the explanatory variables that explain its variability. We study the
effects of these explanatory variables on the response variable using a regression model that
is appropriate for censored data. In this paper, we also introduce a regression model using the BGamma distribution, denoted by $\text{BGamma}(\tetn_{\delta})$ regression model,
for survival times analysis as a feasible alternative to the gamma regression model.
We considered a classic analysis for the $\text{BGamma}(\tetn_{\delta})$ regression model. The inferential part was carried out using the asymptotic distribution of the ML estimators.

The rest of the paper proceeds as follows. In Section \ref{sec:02},
we introduce the bimodal gamma distribution. In Section~\ref{sec:03}, we discuss several mathematical properties of the proposed model. In Section~\ref{sec:04}, we consider likelihood-based methods to estimate the model parameters. In Section~\ref{sec:05}, we carry out a Monte Carlo simulation study to evaluate the performance of the ML estimators. In Section~\ref{sec:06}, we derive a regression model based on the proposed distribution. In Section~\ref{sec:07}, we illustrate the proposed methodologies with two real data sets. Finally, in Section~\ref{sec:08}, we make some concluding remarks.

\section{The bimodal gamma distribution}\label{sec:02}
We say that a random variable $X$ has a BGamma distribution with parameter vector
$\boldsymbol{\theta}_\delta\coloneqq(\alpha,\beta, \delta)$, $\alpha>0,\beta>0$ and $\delta\in\mathbb{R}$,
denoted by $X\sim \text{BGamma}(\boldsymbol{\theta}_\delta)$,
if its PDF is given by
\begin{align}\label{Gamma-density}
    f(x;\boldsymbol{\theta}_\delta)
    =
    \begin{cases}
    \displaystyle
    \frac{1+(1-\delta{x})^2}{
    Z(\boldsymbol{\theta}_\delta) } \, {\beta^\alpha\over \Gamma(\alpha)}\,
    x^{\alpha-1} \, \textrm{e}^{-\beta x},
    & x>0,
    \\
    0, & \text{otherwise},
    \end{cases}
\end{align}
where
$
Z(\boldsymbol{\theta}_\delta)
\coloneqq
2+{\alpha\delta\over\beta}\, [(1+\alpha) {\delta\over\beta}-2)]
$
is the normalization constant,
and
$\Gamma(\alpha)$ is the gamma function. When $\delta=0$, we obtain the classic gamma distribution
with parameter vector $\boldsymbol{\theta}_0=(\alpha,\beta, 0)\coloneqq (\alpha,\beta)$.
Figure \ref{fig:1} shows some different shapes of the $\text{BGamma}$ PDF for different combinations of parameters. This figure reveals clearly the bimodality effect caused by the parameter $\delta$.

If $Y$ is a non-negative random variable following a gamma distribution with  parameter vector $\boldsymbol{\theta}_0$,
denoted by $Y\sim \text{BGamma}(\boldsymbol{\theta}_0)$,
note that in fact the non-negative function $f(\cdot;\boldsymbol{\theta}_\delta)$ is a PDF since
\[
\int_{0}^{\infty}f(x;\boldsymbol{\theta}_\delta)\,\textrm{d}x
=
\frac{1+\mathbb{E}(1-\delta{Y})^2}{Z(\boldsymbol{\theta}_\delta)}
=
\frac{1+{\alpha\delta^2\over\beta^2}+(1-{\alpha\delta\over\beta})^2 }{
Z(\boldsymbol{\theta}_\delta)}
=
1.
\]
%
\begin{proposition}[Monotonicity of the PDF]
The PDF of the \textrm{BGamma} distribution \eqref{Gamma-density}
is decreasing as $\alpha\leqslant 1$, $\delta>0$ and $x<1/\delta$.
\end{proposition}
\begin{proof}
Note that the function
$g(x)\coloneqq1+(1-\delta{x})^2$ is decreasing as $\alpha\leqslant 1$, $\delta>0$ and $x<1/\delta$.
Furthermore, when $\alpha\leqslant 1$, the density
$f(x;\boldsymbol{\theta}_\delta)$ is the product of the function $g(x)$ and a
decreasing and nonnegative function. Thus, the proof is complete.
\end{proof}

\begin{figure}[htb!]
	\begin{center}
		\hspace{0.7cm}(a)\hspace{4.5cm} (b)\\
		\includegraphics[width=5cm,height=4.5cm]{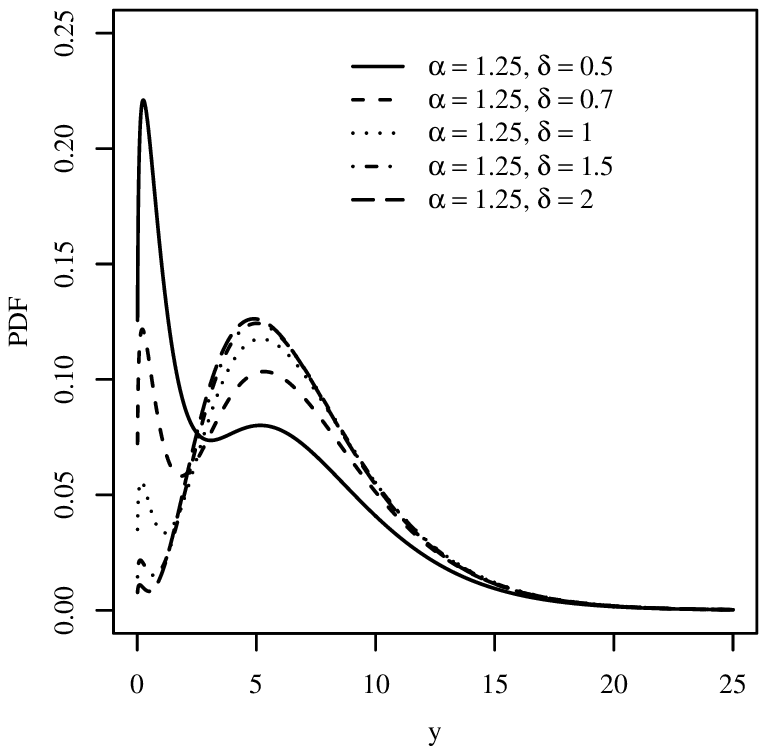}~
		~\includegraphics[width=5cm,height=4.5cm]{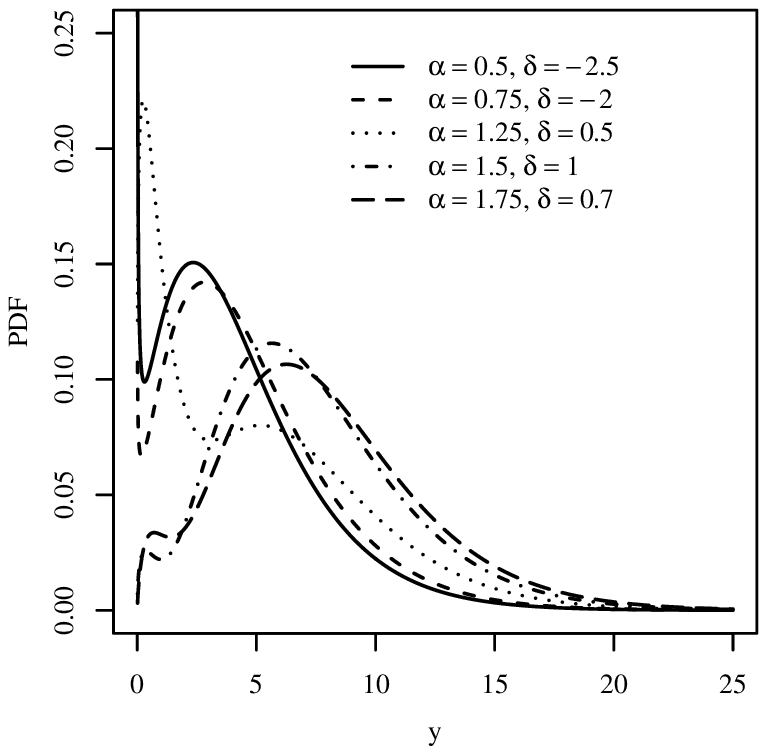}
		\caption{Bimodal gamma PDFs for some parameter values $(\beta=0.50).$}\label{fig:1}
	\end{center}
\end{figure}

\section{Mathematical properties}\label{sec:03}

\subsection{Characterization of unimodality and bimodality}
\begin{proposition}[Modes]\label{mode}
The point $x$ is a mode of the \textrm{BGamma} density \eqref{Gamma-density},
if and only if it is the solution of the following cubic polynomial equation
\[
\big[\beta\delta^2x^2-2\delta(\delta+\beta)x+2(\delta+\beta)\big]x
-
\big[1+(1-\delta x)^2\big]
(\alpha-1)=0,
%
\]
or equivalently
\[
\beta\delta^2 x^3-\delta\big[2(\delta+\beta)+\delta(\alpha-1)\big]x^2
+2\big[\delta+\beta+2(\alpha-1)\big]x-2(\alpha-1)
=0.
\]
\end{proposition}
\begin{proof}
The proof is trivial and omitted.
\end{proof}
\begin{theorem}[Unimodality]\label{bimodal}
The PDF of the \textrm{BGamma} distribution \eqref{Gamma-density} is unimodal in the
following cases:
\begin{enumerate}
\item for $\delta=0$ and $\alpha>1$;
\item for $\delta\geqslant\beta$ and $\alpha=1$.
\end{enumerate}
\end{theorem}
\begin{proof}
(1) When $\delta=0$ and $\alpha>1$  it is well-known that the density \eqref{Gamma-density}
increases and then decreases, with mode at $(\alpha-1)/\beta$.

To prove Item (2) we suppose that
 $x$ is a mode of the \textrm{BGamma} density and that $\alpha=1$. In this case,
the point
$x$ must be the solution of the  quadratic polynomial equation
$p_2(x)\coloneqq \beta\delta^2 x^2-2\delta(\delta+\beta)x+2(\delta+\beta)=0$
(see Proposition \ref{mode}). The discriminant of $p_2$  is given by
$\Delta_2=4\delta^2(\delta+\beta)(\delta-\beta).$

If  $\delta=\beta$ and $\alpha=1$,  $\Delta_2=0$. Then, there is one real zero of
multiplicity two for $p_2(x)=0$,  denoted by $x_0$. Note that $x_0=2/\beta$.
Since $f(x;\boldsymbol{\theta}_\delta)\to 1$ as $x\to 0^+$ and
$f(x;\boldsymbol{\theta}_\delta)\to 0$ as $x\to \infty$, it follows that
the density \eqref{Gamma-density} increases on the interval $(0,x_0)$ and then decreases on
$(x_0,\infty)$. Then $x_0$ is the unique global maximum point.

On the other hand,
if $\delta>\beta$ and $\alpha=1$ note that
$\Delta_2>0$. Then, the equation $p_2(x)=0$ has two distinct rational
zeros, denoted by $x_1, x_2$.
Note that $x_1=(\delta+\beta-\sqrt{\delta^2-\beta^2})/(\beta\delta)>0$ and
$x_2=(\delta+\beta+\sqrt{\delta^2-\beta^2})/(\beta\delta)>0$, and $x_1<x_2$.
Since $f(x;\boldsymbol{\theta}_\delta)\to \beta^3/[\beta^2+\delta(\delta-\beta)]$ as $x\to 0^+$ and
$f(x;\boldsymbol{\theta}_\delta)\to 0$ as $x\to \infty$, it follows that
the BGamma density \eqref{Gamma-density} decreases on the interval $(0,x_1)$, increases
on $(x_1,x_2)$ and then decreases on $(x_2,\infty)$. That is, $x_1$ and $x_2$ are
minimum and maximum points respectively.
\end{proof}
%
%

To state the following result, we define
\begin{align}
& a_{\delta,\beta}
\coloneqq
\delta(4+\delta)(\delta+\beta)+\beta(3\delta-4)(3\delta+4); \label{c2}
\\
& b_{\delta,\beta}
\coloneqq
16(1+\delta)(\delta+\beta)^2+
\big[\delta^2+18\beta\delta(4+\delta)-96\beta\big](\delta+\beta)-2\delta^2;
\label{c3}
\\
& c_{\delta,\beta}\coloneqq
4(4+\delta)(\delta+\beta)^3+
12\beta(3\delta-4)(\delta+\beta)^2-4\delta(\delta+\beta)-27\beta.
\label{c4}
\end{align}
\begin{theorem}[Bimodality and unimodality]\label{bimodal-1}
The PDF of the \textrm{BGamma} distribution \eqref{Gamma-density}, as $\alpha>1$,
has the following shapes.
\begin{enumerate}
\item It is bimodal as  $\delta>\beta$, $a_{\delta,\beta}>0$, $b_{\delta,\beta}>0$ and
$c_{\delta,\beta}>0$. Just take, for example, $\beta=2$ and $\delta=3$;
\item It is unimodal as  $0<\delta<\beta$, $a_{\delta,\beta}<0$, $b_{\delta,\beta}<0$ and
$c_{\delta,\beta}<0$;
\item It is bimodal as  $\delta=\beta>{8\sqrt{3}\over 11}-{4\over 11}$;
\item It is unimodal as $0<\delta=\beta<{\sqrt{1745}\over 12}-{35\over 12}$;
\end{enumerate}
where $a_{\delta,\beta}, b_{\delta,\beta}$ and $c_{\delta,\beta}$
are as in \eqref{c2}, \eqref{c3} and \eqref{c4}, respectively.
\end{theorem}
\begin{proof}
If $x$ is a mode of the \textrm{BGamma} density, by Proposition \ref{mode}
the point $x$ must be the solution of the  cubic polynomial equation
$p_3(x)\coloneqq \beta\delta^2 x^3-\delta\big[2(\delta+\beta)+\delta(\alpha-1)\big]x^2
+2\big[\delta+\beta+2(\alpha-1)\big]x-2(\alpha-1)=0$.
By Descartes' rule of signs (see, e.g. \cite{xue2012loop, griffiths1947introduction}),
$p_3(x)$ has three or one positive roots.
It is well-known that the discriminant of a cubic polynomial $ax^3 + bx^2 + cx + d$ is given by
$\Delta_3 = b^2 c^2 -4ac^3 - 4b^3 d - 27a^2 d^2 + 18abcd$. In our case, we have
\begin{align*}
\Delta_3=\Delta_3(\alpha)
&=
16\delta^4(\alpha-1)^4+16\delta^2 a_{\beta,\delta}\, (\alpha-1)^3
\\
& \quad
+ 4\delta^2 b_{\beta,\delta}\, (\alpha-1)^2
+ 4\delta^2 c_{\beta,\delta}\, (\alpha-1)
+ 16\delta^2(\delta-\beta)(\delta+\beta)^3.
\end{align*}

(1)
Since $\delta>\beta$, $a_{\delta,\beta}$, $b_{\delta,\beta}$ and
$c_{\delta,\beta}$ are positive, we have
$\Delta_3(\alpha)>0$ for each $\alpha>1$. Then the equation $p_3(x)=0$ has three distinct
positive roots, denoted by $x_1, x_2, x_3$.
Let's assume that $x_1<x_2<x_3$.
Since $f(x;\boldsymbol{\theta}_\delta)\to 0$ as $x\to 0^+$ and
$f(x;\boldsymbol{\theta}_\delta)\to 0$ as $x\to \infty$, it follows that
the BGamma density \eqref{Gamma-density} increases on the intervals $(0,x_1)$ and
$(x_2,x_3)$, and decreases on $(x_1,x_2)$ and $(x_3,\infty)$. That is,
$x_1$ and $x_3$ are two  maximum points and $x_1$ is the unique minimum point.

(2) Since $0<\delta<\beta$, $a_{\delta,\beta}$, $b_{\delta,\beta}$ and
$c_{\delta,\beta}$ are negative, it follows that $\Delta_3(\alpha)<0$ for each $\alpha>1$.
Hence, the polynomial equation $p_3(x)=0$ has one positive root, denoted by $x_0$, and two non-real complex
conjugate roots.
Since $f(x;\boldsymbol{\theta}_\delta)\to 0$ as $x\to 0^+$ and
$f(x;\boldsymbol{\theta}_\delta)\to 0$ as $x\to \infty$, note that $x_0$ has to be a maximum point.

To prove Items (3) and (4), note that if $\delta=\beta$,  then
\begin{align*}
& a_{\delta,\beta}=\beta(11\beta^2+8\beta-16);
\\
& b_{\delta,\beta}=\beta^2(18\beta^2+105\beta-65);
\\
& c_{\delta,\beta}=\beta(32\beta^3+272\beta^2-200\beta-27).
\end{align*}
For each $\delta=\beta>{8\sqrt{3}\over 11}-{4\over 11}$ we obtain that
$a_{\delta,\beta},$ $b_{\delta,\beta}$ and $c_{\delta,\beta}$ are
positive quantities, then $\Delta_3(\alpha)>0$ for each $\alpha>1$,
and the proof of Item (3) follows analogously to Item (1).
On the other hand, for $0<\delta=\beta<{\sqrt{1745}\over 12}-{35\over 12}$ note that
$ a_{\delta,\beta},$ $b_{\delta,\beta}$ and $c_{\delta,\beta}$ are
negative. Hence, $\Delta_3(\alpha)<0$ for each $\alpha>1$,
and the proof of Item (4) follows analogously to Item (2).
\end{proof}
\begin{remark}
In the proof of Theorem \ref{bimodal-1}, Item (1), another way to verify that the polynomial
equation
$p_3(x)=\beta\delta^2 x^3-\delta\big[2(\delta+\beta)+\delta(\alpha-1)\big]x^2
+2\big[\delta+\beta+2(\alpha-1)\big]x-2(\alpha-1)=0$ has exactly three
positive roots is to use the Vieta's formula (see, e.g., \cite{vinberg2003course}). Indeed,
in our case the the Vieta's formula is expressed as
\begin{align*}
x_1+x_2+x_3={2(\delta+\beta)+\delta(\alpha-1)\over \beta\delta},
\\
x_1\, x_2+ x_1\, x_3+ x_2\, x_3={2[\delta+\beta+2(\alpha-1)]\over \beta\delta^2},
\\
x_1\, x_2\, x_3= {2(\alpha-1)\over \beta\delta^2}.
\end{align*}
From the above equations the claim follows.
\end{remark}
%
%

%
%
%
%
%

\subsection{Real moments, variance and moment generating function}
The following result shows that the existence of the classic gamma moments is inherited
for the \textrm{BGamma} distribution.
\begin{proposition}[Moments]\label{prop-1}
If $X\sim \text{BGamma}(\boldsymbol{\theta}_\delta)$, for each fixed real number
$\nu$ such that $\nu>-\alpha$, we have
\begin{align*}
\mathbb{E}X^\nu
=
\frac{
\big[
2 - 2 {\delta\over \beta} (\nu+\alpha) + {\delta^2\over \beta^2} (\nu+\alpha+1)
\big] \Gamma(\nu+\alpha)
}{Z(\boldsymbol{\theta}_\delta) \beta^\nu \Gamma(\alpha)}.
\end{align*}
\begin{proof}
A straightforward computation shows that
\[
\mathbb{E}X^\nu
=
\frac{2\,\mathbb{E}Y^\nu-2\delta\,\mathbb{E}Y^{\nu+1} +\delta^2\,\mathbb{E}Y^{\nu+2} }{
    Z(\boldsymbol{\theta}_\delta)}, \quad Y\sim \text{BGamma}(\boldsymbol{\theta}_0).
\]
Since $\mathbb{E}Y^\nu={\Gamma(\nu+\alpha)\over \beta^\nu \Gamma(\alpha)}$,
$\nu>-\alpha,$ and $\Gamma(x+1)=x\Gamma(x)$, the proof follows.
\end{proof}
\end{proposition}
\begin{corollary}[Mean and variance]\label{def-mu-sigma}
Set $\mu\coloneqq\mathbb{E}X$ and $\sigma\coloneqq\sqrt{{\rm Var}(X)}$.
By Proposition \ref{prop-1}, it immediately follows that
\begin{align*}
&\mu=
{\alpha\over \beta Z(\boldsymbol{\theta}_\delta)}\,
\kappa(\boldsymbol{\theta}_\delta),
\\
&\sigma^2
=
%
{\alpha\over \beta^2 Z(\boldsymbol{\theta}_\delta)}
\left[
\alpha\, \kappa^2(\boldsymbol{\theta}_\delta)
+
(1+\alpha)
Z(\boldsymbol{\theta}_\delta)\,
\kappa(\boldsymbol{\theta}_\delta)
+
{\delta\over\beta}\Big({\delta\over\beta}-2\Big)(1+\alpha) Z(\boldsymbol{\theta}_\delta)
\right],
\end{align*}
where
$\kappa(\boldsymbol{\theta}_\delta)
\coloneqq
2 - 2 {\delta\over \beta} (1+\alpha) + {\delta^2\over \beta^2} (2+\alpha)$.
\end{corollary}
%
%
\begin{proposition}[Standardized moments]\label{prop-2}
    If $X\sim \text{BGamma}(\boldsymbol{\theta}_\delta)$, for each fixed natural number
    $n$ we have
    \begin{align*}
    \mathbb{E}\biggl({X-\mu \over \sigma}\biggr)^n
    =
    {1\over\sigma^n}
    \sum_{k=0}^{n}\binom{n}{k}(-\mu)^{n-k}\,
    \frac{
        2 - 2 {\delta\over \beta} (k+\alpha) + {\delta^2\over \beta^2} (k+\alpha+1)
    }{
    Z(\boldsymbol{\theta}_\delta) \beta^k}\, \prod_{i=0}^{k-1}(\alpha+i),
\end{align*}
where $\mu$ and $\sigma$ is as in Corollary \ref{def-mu-sigma}, and
$\prod_{i=0}^{-1}(\alpha+i)\coloneqq1$.
In particular, by taking $n = 3$ and $n = 4$ we have closed expressions for the
skewness and kurtosis of $X$, respectively.
\end{proposition}
\begin{proof}
The proof of this proposition follows immediately by combining the
Binomial expansion with the Proposition \ref{prop-1} and with the identity
$\Gamma(x+1)=x\Gamma(x)$.
\end{proof}
\begin{proposition}\label{exp-log}
    If $X\sim \text{BGamma}(\boldsymbol{\theta}_\delta)$, for each fixed natural number
    $n$ we have
\begin{enumerate}
\item
$
\mathbb{E}\log X^n
=
{
{n\delta\over\beta}[{(2\alpha+1)\delta\over\beta}-2]
+
n[2-{2n\alpha\delta\over\beta}+{\alpha(\alpha+1)\delta^2\over\beta^2}]
[\Psi^{(0)}(\alpha)-\log\beta]
\over
Z(\boldsymbol{\theta}_\delta)
}
$;
\item $\mathbb{E}(\log X)^n
=
{
{1\over \Gamma(\alpha)}
\{
{n(n-1)\delta^2\over\beta^2} - {n\delta\over\beta}{2+[{\alpha(\alpha+1)\delta^2\over\beta}]}
\}
\sum_{k=0}^{n-2}\binom{n-2}{k}(-1)^{n-2-k}(\log \beta)^{n-2-k}\Psi^{(k)}(\alpha)
\over
Z(\boldsymbol{\theta}_\delta)
}
$

$
\qquad\qquad+
{
{1\over\Gamma(\alpha)}
\{2+[{\alpha(\alpha+1)\delta^2\over\beta}]\}[{\Psi^{(n)}(\alpha)}-
(n-1){\log\beta\, \Psi^{(n-1)}(\alpha)}]
-
{n\delta\over\beta \Gamma(\alpha)}
[2+{(2\alpha+1)\delta\over\beta}] {\Psi^{(n-1)}(\alpha)}
\over
Z(\boldsymbol{\theta}_\delta)
}
$,
\end{enumerate}
where $\Psi^{(m)}(z)$ is the polygamma function of order $m$ defined by
${{\rm d}^{m+1}\over {\rm d}z^{m+1}}\log \Gamma(z)$.
\end{proposition}
\begin{proof}
Let $Y\sim \text{BGamma}(\boldsymbol{\theta}_0).$
Integration by parts gives
\begin{align*}
&
\mathbb{E}Y\log Y^n
=
{n\over\beta}
+
{\alpha\over\beta}\,
\mathbb{E}\log Y^n,
\\
&
\mathbb{E}Y^2\log Y^n
=
{n(\alpha+1)\over\beta^2}
+
{n\over\beta}\,
\mathbb{E}Y
+
{\alpha(\alpha+1)\over\beta^2}\,
\mathbb{E}\log Y^n.
\end{align*}
Since
\begin{align*}
\mathbb{E}\log X^n
=
{
2\,\mathbb{E}\log Y^n - 2\delta\,\mathbb{E}Y\log Y^n+\delta^2\,\mathbb{E}Y^2\log Y^n
\over
Z(\boldsymbol{\theta}_\delta)
}
\end{align*}
and $\mathbb{E}\log Y^n=n(\Psi^{(0)}(\alpha)-\log\beta)$, by combining the above identities
with Proposition \ref{prop-1},
the proof of first item follows.

On the other hand, to prove Item (2), note that integration by parts gives
\begin{align*}
&
\mathbb{E}Y(\log Y)^n={n\over\beta}\, \mathbb{E}(\log Y)^{n-1}
+{\alpha\over\beta}\, \mathbb{E}(\log Y)^{n},
\\
&
\mathbb{E}Y^2(\log Y)^n
\!=
{n(n-1)\over\beta^2}\, \mathbb{E}(\log Y)^{n-2}
\!+
{n(2\alpha+1)\over\beta^2}\, \mathbb{E}(\log Y)^{n-1}
\!+
{\alpha(\alpha+1)\over\beta^2}\, \mathbb{E}(\log Y)^{n}.
\end{align*}
Since
\begin{align*}
&\mathbb{E}(\log X)^n
=
{
2\,\mathbb{E}(\log Y)^n - 2\delta\,\mathbb{E}Y(\log Y)^n+\delta^2\,\mathbb{E}Y^2(\log Y)^n
\over
Z(\boldsymbol{\theta}_\delta)
},
\\
&
\mathbb{E}(\log Y)^n
=
{1\over \Gamma(\alpha)}
\sum_{k=0}^{n}\binom{n}{k}(-1)^{n-k}(\log \beta)^{n-k}\Psi^{(k)}(\alpha),
\end{align*}
and
\begin{align*}
&\mathbb{E}(\log Y)^{n-1}
=
{\Psi^{(n-1)}(\alpha)\over\Gamma(\alpha)}
+
\mathbb{E}(\log Y)^{n-2},
\\
&
\mathbb{E}(\log Y)^n
=
{\Psi^{(n)}(\alpha)\over\Gamma(\alpha)}
-
(n-1)\, {\log\beta\, \Psi^{(n-1)}(\alpha)\over\Gamma(\alpha)}
+
\mathbb{E}(\log Y)^{n-2},
\end{align*}
by combining the above identities the proof follows.
\end{proof}

Let $M_X(t)\coloneqq\mathbb{E}\textrm{e}^{tX}$ be the moment generating function of $X$
(if it exists).
The known identity (see \cite{jkk:93})
$M_X(t)=\sum_{r=0}^{\infty}{ t^r \mathbb{E}X^r\over r!}$, whenever it exists,
simply provides an expression for the moment generating function of $X$ since the moments
of $X$ exist (see Proposition \ref{prop-1}).
The following result gives us a closed expression for this function.
\begin{proposition}\label{prop-gm}
If $X\sim \text{BGamma}(\boldsymbol{\theta}_\delta)$  then
\[
M_X(t)=
{2+\delta^2 \alpha(\alpha+1)-2\delta\alpha(\beta-t)
\over
Z(\boldsymbol{\theta}_\delta) \beta^{-\alpha}}\,
(\beta-t)^{-(\alpha+2)},
\quad \text{for} \ t<\beta.
\]
\end{proposition}
\begin{proof}
Let $Y\sim \text{BGamma}(\boldsymbol{\theta}_0).$
For $t<\beta$, integration by parts gives
\begin{align*}
\mathbb{E}Y\textrm{e}^{tY}
=
{\alpha\over\beta-t}\,  M_Y(t),
\quad
\mathbb{E}Y^2\textrm{e}^{tY}
=
{\alpha(\alpha+1)\over(\beta-t)^2}\, M_Y(t).
\end{align*}
Since
\[
M_X(t)=2\,M_Y(t)-2\delta\,\mathbb{E}Y\textrm{e}^{tY}+\delta^2\, \mathbb{E}Y^2\textrm{e}^{tY},
\]
combining the above identities, we obtain
\[
M_X(t)={2+\delta^2 \alpha(\alpha+1)-2\delta\alpha(\beta-t)
    \over
    Z(\boldsymbol{\theta}_\delta) (\beta-t)^2}\,
M_Y(t),
\quad \text{for} \ t<\beta.
\]
Since $M_Y(t)=(1-{t\over\beta})^{-\alpha}$ for $t<\beta$, the proof follows.
\end{proof}
\begin{remark}
The characteristic function of $X\sim \text{BGamma}(\boldsymbol{\theta})$, denoted by $\phi_X(t)$, can
be obtained from the moment generating function by the relation $M_X(t)=\phi_X(-it)$.
\end{remark}
The next result shows that the tail of the \text{BGamma} distribution \eqref{Gamma-density}
function decays to zero exponentially or faster.
\begin{corollary}[Light-tailed distribution]
    If $X\sim {\rm BGamma}(\boldsymbol{\theta}_\delta)$, then
    there exists $t>0$ such that $\mathbb{P}(X>x)\leqslant \text{e}^{-tx}$ for $x$ large enough.
\end{corollary}
\begin{proof}
    Since, by Proposition \ref{prop-gm}, there exists $t<\beta$ such that $M_X(t)<\infty$,
    $X\sim \text{BGamma}(\boldsymbol{\theta}_\delta)$, the proof follows.
\end{proof}
\begin{remark} Let $X$ an absolutely continuous random variable with density function $f_X(\cdot)$.
Following \cite{embrechts_1998}, the rate of a random variable is
\[
\tau_X\coloneqq -\lim_{x\to\infty}{{\rm d} \ln f_X(x)\over {\rm d}x}.
\]
Note that
\begin{align*}
\tau_{{\rm BGamma}(\boldsymbol{\theta}_\delta)}
&=
\lim_{x\to\infty}
\left[{2\delta(1-\delta x)\over 1+(1-\delta x)^2}-(\alpha-1){1\over x}+\beta\right]
=
\beta
\\[0,1cm]
&=
\tau_{{\rm BGamma}(\boldsymbol{\theta}_0)}
=
\tau_{{\rm BGamma}(\alpha=1,\beta,\delta=0)}
=
\tau_{{\rm exp}(\beta)}.
\end{align*}
That is, the rate of a BGamma-distributed random variable depends only on its scale $\beta$.
In other words, far enough out in the tail, every BGamma distribution looks like an
exponential distribution.
On the other hand, it is simple to verify that
\begin{align*}
\tau_{{\rm InvGamma}(\boldsymbol{\theta}_0)}
=
\tau_{{\rm LogNorm}(\mu,\sigma^2)}
=
\tau_{{\rm GenPareto}(\boldsymbol{\theta}_0, \xi)}
= 0
<
\tau_{{\rm BGamma}(\boldsymbol{\theta}_\delta)}
<
\tau_{{\rm Normal}(\mu,\sigma^2)}=\infty.
\end{align*}
Therefore, the tail of the normal distribution is lighter than the tail of the BGamma
distribution, which is lighter than the tails of the generalized-Pareto, log-normal, and inverse-gamma
distributions.
\end{remark}

\subsection{Reliability, hazard rate and the mean residual life}

For each $t\geqslant 0$,
the reliability, the hazard rate and the mean residual life functions are defined as
\begin{align*}
R(t;&\boldsymbol{\theta}_\delta)
\coloneqq \int_{t}^{\infty}
f(x;\boldsymbol{\theta}_\delta)
\, \textrm{d} x,
\quad
H(t;\boldsymbol{\theta}_\delta)
\coloneqq
{f(t;\boldsymbol{\theta}_\delta)\over R(t;\boldsymbol{\theta}_\delta)},
\\
& \textrm{MRL}(t;\boldsymbol{\theta}_\delta)
\coloneqq
{1\over R(t;\boldsymbol{\theta}_\delta)}
\int_{t}^{\infty}
R(x;\boldsymbol{\theta}_\delta)
\, \textrm{d} x,
\end{align*}
respectively.

Let $Y\sim \text{BGamma}(\boldsymbol{\theta}_0).$
Integration by parts gives
\begin{align}
&
\mathbb{E}\mathds{1}_{\{Y\geqslant t\}} Y
=
{\textrm{e}^{-\beta t}\over\beta} t^{\alpha}
+
{\alpha\over\beta}\,
\mathbb{E}\mathds{1}_{\{Y\geqslant t\}},  \label{id-1}
\\
&
\mathbb{E}\mathds{1}_{\{Y\geqslant t\}} Y^2
=
{\textrm{e}^{-\beta t}\over\beta}t^\alpha\Big(t+{\alpha+1\over\beta} \Big)
+
{\alpha(\alpha+1)\over\beta^2}\,
\mathbb{E}\mathds{1}_{\{Y\geqslant t\}},  \label{id-2}
\\
&
\mathbb{E}\mathds{1}_{\{Y\geqslant t\}} Y^3
=
{\textrm{e}^{-\beta t}\over\beta} t^\alpha
\Big[ t^2+{\alpha+2\over\beta}t +
{(\alpha+1)(\alpha+2)\over\beta^2} \Big]                                     \label{id-3}
    \\
&\hspace{2.5cm}
+
{\alpha(\alpha+1)(\alpha+2)\over\beta^3}\,
\mathbb{E}\mathds{1}_{\{Y\geqslant t\}}.  \nonumber
\end{align}
\begin{proposition}\label{prop-3}
    If $X\sim \text{BGamma}(\boldsymbol{\theta}_\delta)$ then
    \begin{enumerate}
    \item Reliability function:
    $
    R(t;\boldsymbol{\theta}_\delta)
    =
    {{\delta t\over\beta}[\delta(t+{\alpha+1\over \beta})-2] \over Z(\boldsymbol{\theta}_\delta)} \,
    f(t;\boldsymbol{\theta}_0)
    +
    R(t;\boldsymbol{\theta}_0);
    $
\item Cumulative distribution function:
$
    F(t;\boldsymbol{\theta}_\delta)
    =
    -{{\delta t\over\beta}[\delta(t+{\alpha+1\over \beta})-2] \over Z(\boldsymbol{\theta}_\delta)} \,
    f(t;\boldsymbol{\theta}_0)
    +
    F(t;\boldsymbol{\theta}_0);
$
\item Hazard rate:
$
H(t;\boldsymbol{\theta}_\delta)
=
{
    [1+(1-\delta{t})^2] \,H(t;\boldsymbol{\theta}_0)
    \over
    {\delta t\over\beta}[\delta(t+{\alpha+1\over \beta})-2]\, H(t;\boldsymbol{\theta}_0)
    +
    Z(\boldsymbol{\theta}_\delta)
};
$
    \end{enumerate}
where
$R(t;\boldsymbol{\theta}_0)
=
1-F(t;\boldsymbol{\theta}_0)
=
\mathbb{E}\mathds{1}_{\{Y\geqslant t\}}
=
{\beta^\alpha\over\Gamma(\alpha)}
\int_{t}^{\infty}y^{\alpha-1} \textrm{e}^{-\beta y}\, {\rm d}y$.
\end{proposition}
\begin{proof}
Since
\begin{align*}
R(t;\boldsymbol{\theta}_\delta)
=
2 \, \mathbb{E}\mathds{1}_{\{Y\geqslant t\}}
-
2\delta \, \mathbb{E}\mathds{1}_{\{Y\geqslant t\}} Y
+
\delta^2 \, \mathbb{E}\mathds{1}_{\{Y\geqslant t\}} Y^2,
\quad
Y\sim \text{BGamma}(\boldsymbol{\theta}_0),
\end{align*}
using the identities \eqref{id-1} and \eqref{id-2}, the proof of Item (1) follows.
The proof of items (2) and (3) follows directly by combining the definitions of
$F(t;\boldsymbol{\theta}_\delta)$ and $H(t;\boldsymbol{\theta}_\delta)$ with Item (1),
respectively.
\end{proof}

Figure~\ref{fig:2} shows some different shapes of the BGamma hazard rate for different combinations of parameters.

\begin{figure}[htb!]
	\begin{center}
		\hspace{0.7cm}(a)\hspace{4.5cm} (b)\\
		\includegraphics[width=5cm,height=4.5cm]{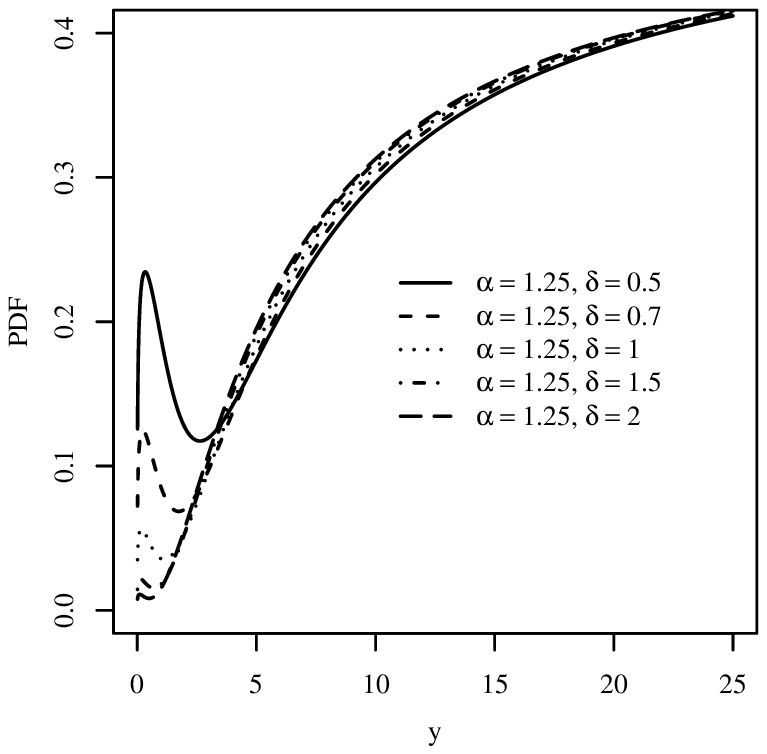}~
		~\includegraphics[width=5cm,height=4.5cm]{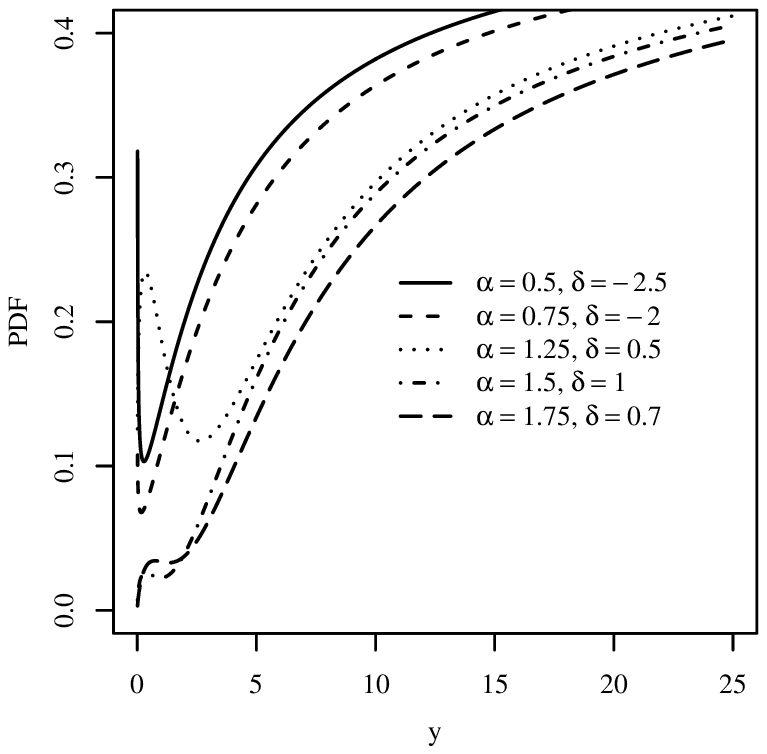}
		\caption{BGamma hazard rate for some parameter values ($\beta=0.50$).\label{fig:2}}
	\end{center}
\end{figure}

\begin{remark}[Monotonicity of the hazard function when $\delta=0$]\label{rem-hazard}
It is well-known that,
when $\alpha>1$, the hazard function $H(t;\boldsymbol{\theta}_0)$ is
concave and increasing. When $\alpha<1$,
the hazard function is convex and decreasing. The case $\alpha=1$ corresponds to the
exponential distribution which has constant hazard function.
\end{remark}
%
%
\begin{proposition}[Decreasing monotonicity of the hazard rate]
The hazard rate $H(x;\boldsymbol{\theta}_\delta)$ of the \textrm{BGamma} distribution \eqref{Gamma-density}
is decreasing when $\alpha\leqslant 1$, $\delta>0$ and
$x\in({1\over\delta}-{\alpha+1\over 2\beta},\, {1\over\delta})$.
\end{proposition}
\begin{proof}
By Proposition \ref{prop-3},
\begin{align}\label{exp-hazard}
H(x;\boldsymbol{\theta}_\delta)
=
{
    1+(1-\delta{x})^2
    \over
    {\delta x\over\beta}\big[\delta(x+{\alpha+1\over \beta})-2\big]
    +
    {Z(\boldsymbol{\theta}_\delta)\over H(x;\boldsymbol{\theta}_0)}
}.
\end{align}
A straightforward computation shows that the function
$x\mapsto 1+(1-\delta{x})^2$ decreases when $x<1/\delta$ and that the function
$x\mapsto {\delta x\over\beta}[\delta(x+{\alpha+1\over \beta})-2]$
increases when $x>{1\over\delta}-{\alpha+1\over 2\beta}$. Then, using
Remark \ref{rem-hazard},
$x\mapsto   {\delta x\over\beta}[\delta(x+{\alpha+1\over \beta})-2]
+
{Z(\boldsymbol{\theta}_\delta)\over H(x;\boldsymbol{\theta}_0)}$ is an increasing function
when $\alpha\leqslant 1$ and $x>{1\over\delta}-{\alpha+1\over 2\beta}$.
Hence, by expression \eqref{exp-hazard}, the hazard rate
is the product of two decreasing and nonnegative functions
when $\alpha\leqslant 1$, $\delta>0$ and
$x\in({1\over\delta}-{\alpha+1\over 2\beta},\, {1\over\delta})$. The proof follows.
\end{proof}
\begin{proposition}[Increasing monotonicity of the hazard rate]
The hazard rate $H(x;\boldsymbol{\theta}_\delta)$ of the \textrm{BGamma} distribution \eqref{Gamma-density}
is increasing in the following cases:
\begin{enumerate}
    \item for $\delta=\beta$, $\alpha=1$ and
    $x\in (0,2/\beta)$;
    \item for $\delta>\beta$, $\alpha=1$ and
    $x\in \big(\delta+\beta-\sqrt{\delta^2-\beta^2})/(\beta\delta),\, \delta+\beta+\sqrt{\delta^2-\beta^2})/(\beta\delta)\big)$;
    \item under the conditions $\alpha>1$, $\delta>\beta$, $a_{\delta,\beta}>0$, $b_{\delta,\beta}>0$ and
    $c_{\delta,\beta}>0$, for
    $x\in(0,x_1)$ or $x\in(x_2,x_3)$, where $x_1,x_2,x_3$ are the three distinct positive
    roots of the polynomial equation $p_3(x)=0$;
    \item under the conditions  $\alpha>1$, $0<\delta<\beta$, $a_{\delta,\beta}<0$, $b_{\delta,\beta}<0$ and $c_{\delta,\beta}<0$ and $x\in(0,x_0)$, where $x_0$ is the unique positive
    root of $p_3(x)=0$;
\end{enumerate}
where $p_3(x)=\beta\delta^2 x^3-\delta\big[2(\delta+\beta)+\delta(\alpha-1)\big]x^2
+2\big[\delta+\beta+2(\alpha-1)\big]x-2(\alpha-1)$, and
$a_{\delta,\beta}, b_{\delta,\beta}$ and $c_{\delta,\beta}$
are as in \eqref{c2}, \eqref{c3} and \eqref{c4}, respectively.
\end{proposition}
\begin{proof}
As a sub-product of the proof of Theorems \ref{bimodal} and \ref{bimodal-1}, note that
the density $f(x;\boldsymbol{\theta}_\delta)$ is increasing on the above mentioned intervals.
Since
$H(x;\boldsymbol{\theta}_\delta)
=
{f(x;\boldsymbol{\theta}_\delta)\over R(x;\boldsymbol{\theta}_\delta)}
$
and $R(x;\boldsymbol{\theta}_\delta)$ is a decreasing function, in this case, we have that
the hazard rate function is the product of the two
increasing and nonnegative functions, then the proof of Items (1)-(4) follows.
\end{proof}
\begin{proposition}\label{prop-4}
    If $X\sim \text{BGamma}(\boldsymbol{\theta}_\delta)$ then
\begin{multline*}
\mathbb{E}\mathds{1}_{\{X\geqslant t\}} X
    =
    \frac{
        \big\{2+\delta^2[t^2+{\alpha+2\over\beta}t+{(\alpha+1)(\alpha+2)\over \beta^2}]-2\delta(t+{\alpha+1\over \beta})\big\}
        \beta^{\alpha-1}  t^{\alpha} \, {\rm e}^{-\beta t}
    }
{Z(\boldsymbol{\theta}_\delta) \Gamma(\alpha)}
\\
    +
    {\alpha\over\beta}\,
    \frac{
    \big[2+{\delta^2\over\beta^2}(\alpha+1)(\alpha+2)-2{\delta\over\beta} (\alpha+1)\big]
    R(t;\boldsymbol{\theta}_0)
    }{
    Z(\boldsymbol{\theta}_\delta)}.
\end{multline*}
\end{proposition}
\begin{proof}
Since
\begin{align*}
\mathbb{E}\mathds{1}_{\{X\geqslant t\}} X
=
2\, \mathbb{E}\mathds{1}_{\{Y\geqslant t\}} Y
-
2\delta\, \mathbb{E}\mathds{1}_{\{Y\geqslant t\}} Y^2
+
\delta^2\, \mathbb{E}\mathds{1}_{\{Y\geqslant t\}} Y^3,
\quad
Y\sim \text{BGamma}(\boldsymbol{\theta}_0),
\end{align*}
using the identities \eqref{id-1}, \eqref{id-2} and \eqref{id-3}, the proof follows.
\end{proof}
\begin{remark}[Mean residual life function]\label{rem-1}
Integration by parts gives
\[
\mathbb{E}\mathds{1}_{\{X\geqslant t\}} X
=
t R(t;\boldsymbol{\theta}_\delta)
+
\int_t^{\infty}
R(x;\boldsymbol{\theta}_\delta)
\, {\rm d} x,
\]
since $x R(x;\boldsymbol{\theta}_\delta)\to 0$ as $x\to\infty$. Then
\[
\textrm{MRL}(t;\boldsymbol{\theta}_\delta)
=
\left[
{1\over R(t;\boldsymbol{\theta}_\delta)}\,
\mathbb{E}\mathds{1}_{\{X\geqslant t\}} X
\right]
- t,
\]
where $R(t;\boldsymbol{\theta}_\delta)$ and
$\mathbb{E}\mathds{1}_{\{X\geqslant t\}} X$
are given in Propositions \ref{prop-3} and \ref{prop-4}, respectively.
\end{remark}
\begin{remark}
In the particular case $\delta=0,$ note that
\[
\mathbb{E}\mathds{1}_{\{Y\geqslant t\}} Y
=
{\beta^{\alpha-1}\over\Gamma(\alpha)}\, t^\alpha {\rm e}^{-\beta t}+ \Big({\alpha\over\beta}\Big)
R(t;\boldsymbol{\theta}_0),
\quad
Y\sim \text{BGamma}(\boldsymbol{\theta}_0).
\]
Then, by Remark \ref{rem-1},
\[
\textrm{MRL}(t;\boldsymbol{\theta}_0)
=
{\beta^{\alpha-1} t^\alpha {\rm e}^{-\beta t}\over\Gamma(\alpha) R(t;\boldsymbol{\theta}_0)}
+ \Big({\alpha\over\beta}\Big)
- t.
\]
The above identity was also verified by \cite{GOVIL198347}, Equation (10).
\end{remark}

\subsection{Entropy measures}

Entropy represents the amount of uncertainty of a probability distribution. Some of this measures are particular cases of Renyi's entropy, such as Shannon entropy and Quadratic entropy. \cite{Ambedkar2006}
\color{black}

Let $X\sim \text{BGamma}(\boldsymbol{\theta}_\delta)$.
The Renyi's entropy measure is defined as
\[
H_\gamma(X)
\coloneqq
-{1\over 1-\gamma}\log\int_{0}^{\infty} f^\gamma(x;\boldsymbol{\theta}_\delta) \, \textrm{d}x,
\quad \gamma\geqslant 0 \ \text{and} \ \gamma\neq 1,
\]
and for the quadratic entropy
\[
H_2(X)
\coloneqq
-\log\int_{0}^{\infty} f^2(x;\boldsymbol{\theta}_\delta) \, \textrm{d}x.
\]

We also define the Shannon entropy as
\[
H_1(X)\coloneqq-\int_{0}^{\infty} f(x;\boldsymbol{\theta}_\delta)
\log f(x;\boldsymbol{\theta}_\delta)\, \textrm{d}x.
\]
\begin{proposition}[Quadratic entropy]
    If $X\sim \text{BGamma}(\boldsymbol{\theta}_\delta)$ and $\alpha>1$, then
\begin{align*}
H_2(X)
&=
\log 2+{1\over 2}\log\pi + 2\log\Gamma(\alpha)
+
\log Z(\boldsymbol{\theta}_\delta)
-
\log\big[1+\delta^2\sigma^2+(1-\delta\mu)^2\big]
\\
& \quad
-\log\beta-\log(\alpha-1)-\log\Gamma(\alpha-1)-\log\Gamma\big(\alpha-{1\over2}\big),
\end{align*}
where $\mu$ and $\sigma^2$ are as in Corollary \ref{def-mu-sigma}.
\end{proposition}
\begin{proof}
A straightforward computation shows that
\[
\int_{0}^{\infty} f^2(x;\boldsymbol{\theta}_\delta) \, \textrm{d}x
=
{\beta\,\Gamma(2\alpha-1)
\over
2^{2\alpha-1}
Z(\boldsymbol{\theta}_\delta)
\Gamma^2(\alpha)
}\,
\big[1+ \mathbb{E}(1-\delta X)^2 \big],
\]
where $X\sim \text{BGamma}(2\alpha-1,2\beta,\delta)$ and
$\mathbb{E}(1-\delta X)^2=\delta^2\sigma^2+(1-\delta\mu)^2$.

Combining the formulas $\Gamma(2\alpha+1)=2\alpha\Gamma(2\alpha)$ and
$\Gamma(2z)={2^{2z-1}\over\sqrt{\pi}}\, \Gamma(z)\Gamma(z+{1\over 2})$,
the expression of the right hand can be written as
\[
=
{\beta(\alpha-1)\Gamma(\alpha-1)\Gamma\big(\alpha-{1\over2}\big)
\over
2
Z(\boldsymbol{\theta}_\delta)
\sqrt{\pi} \Gamma^2(\alpha)} \,
\big[1+ \delta^2\sigma^2+(1-\delta\mu)^2 \big].
\]
Finally, taking logarithm and multiplying by $-1$ on both sides of the above identity, we
complete the proof.
\end{proof}
\begin{proposition}[Shannon entropy]
Let $X\sim \text{BGamma}(\boldsymbol{\theta}_\delta)$. The Shannon entropy is given by
\begin{align*}
H_1(X)&=
\log Z(\boldsymbol{\theta}_\delta)+\log \Gamma(\alpha)-\alpha\log \beta
-
{1\over Z(\boldsymbol{\theta}_\delta)}\, \phi(s)|_{s=1 }
\\
&\quad
-
(\alpha-1)\,
{
    {\delta\over\beta}[{(2\alpha+1)\delta\over\beta}-2]
    +
    [2-{2\delta\over\beta}+{\alpha(\alpha+1)\delta^2\over\beta^2}](\Psi^{(0)}(\alpha)-\log\beta)
    \over
    Z(\boldsymbol{\theta}_\delta)
}
+
\beta\mu,
\end{align*}
where $\phi(s)\coloneqq{{{\rm d}\over {\rm d}s}} \mathbb{E}\big[1+(1-\delta Y)^2\big]^s$,
$Y\sim \text{BGamma}(\boldsymbol{\theta}_0)$,
and
$\phi(s)|_{s=1}$ exists. Here, $\mu$ is as in Corollary \ref{def-mu-sigma}
and $\Psi^{(m)}(z)$ is the polygamma function of order $m$ defined by
${{\rm d}^{m+1}\over {\rm d}z^{m+1}}\log \Gamma(z)$.
\end{proposition}
\begin{proof}
    Note that the Shannon entropy can be rewritten as
\[
H_1(X)=\log Z(\boldsymbol{\theta}_\delta)+\log \Gamma(\alpha)-\alpha\log \beta - \mathbb{E}\log g(X) -(\alpha-1)\mathbb{E}\log X+\beta\mathbb{E}X,
\]
where
$Z(\boldsymbol{\theta}_\delta)=
2+{\alpha\delta\over\beta}\, [(1+\alpha) {\delta\over\beta}-2)]
$
and $g(x)\coloneqq 1+(1-\delta x)^2$.
The expectation $\mathbb{E}\log X$ was obtained in Proposition \ref{exp-log} and
$\mathbb{E}X=\mu$ is as in Corollary \ref{def-mu-sigma}.
By \cite{Teh:2006:CVB:2976456.2976626}, we can approximate the function $\log g(x)$
using a second-order Taylor expansion about $\mathbb{E}g(X)$
and evaluate its expectation as follows
\begin{align*}
\mathbb{E}\log g(X)
\approx \log \mathbb{E}g(X) -
{\mathrm{Var}[g(X)]\over 2\mathbb{E}g^2(X)}.
\end{align*}
Since $\mathbb{E}X^\nu<\infty$ for each $\nu>-\alpha$ (see Proposition \ref{prop-1}) we have that
$\log \mathbb{E}g(X)<\infty$  and $\mathbb{E}g^2(X)<\infty$.
Then,
$\mathbb{E}\log g(X)$ exists.
Finally, since
\begin{align*}
\mathbb{E}\log g(X)
&=
{1\over Z(\boldsymbol{\theta}_\delta)}\, \mathbb{E}\big[g(Y)^s \log g(Y)\big]\big|_{s=1},
\quad Y\sim \text{BGamma}(\boldsymbol{\theta}_0)
\\
&=
{1\over Z(\boldsymbol{\theta}_\delta)}\, \mathbb{E}\big[{{\rm d}\over {\rm d}s}g(Y)^s\big]\big|_{s=1}
=
{1\over Z(\boldsymbol{\theta}_\delta)}\, \phi(s)|_{s=1},
\end{align*}
the proof follows.
\end{proof}
%

\section{Maximum likelihood estimation}\label{sec:04}

Let $X$ be a random variable with \text{BGamma} distribution
$f(x;\boldsymbol{\theta}_\delta)$ that depends on a
parameter vector $\boldsymbol{\theta}_\delta=(\alpha,\beta,\delta)$
and let $(X_1,\ldots, X_n)$ be a random sample of $X$ (i.e., the random
variables $X_1,\ldots, X_n$ are independent and identically distributed with
\text{BGamma} distribution) for $\boldsymbol{\theta}_\delta$ in an open subset
(parameter space) $\Theta$ of $\mathbb{R}^3$, where distinct values of $\boldsymbol{\theta}_\delta$
yield distinct distributions for $X_1$.
Denoting ${\bf x}= (x_1,\ldots, x_n)$ as the corresponding observed values of the random sample
$(X_1,\ldots, X_n)$, the log-likelihood function for $\boldsymbol{\theta}_\delta$ is given by
\begin{align}
l(\boldsymbol{\theta}_\delta;{\bf x})
&=
-\log Z(\boldsymbol{\theta}_\delta)
+
\sum_{i=1}^{n}\log\big[1+(1-\delta{x}_i)^2\big]
\\
& \quad +
\alpha\log\beta
-
\log\Gamma(\alpha)
+
(\alpha-1)\sum_{i=1}^{n}\log x_i
-
n\beta\overline{x}, \nonumber
\end{align}
where $Z(\boldsymbol{\theta}_\delta)= 2+{\alpha\delta\over\beta}\, [(1+\alpha) {\delta\over\beta}-2]$ and $\overline{x}={1\over n}\sum_{i=1}^{n}x_i$.
The first-order partial derivatives and the second-order (and mixed) partial derivatives of $Z(\boldsymbol{\theta}_\delta)$ are given by
\begin{align} \label{der-1-2}
\begin{array}{lllll}
&{\partial Z(\boldsymbol{\theta}_\delta)\over \partial \alpha}
=
{\delta\over\beta} \, \big[(1+2\alpha){\delta\over\beta}  -2\big],
&
{\partial Z(\boldsymbol{\theta}_\delta)\over \partial \beta}
=
-{2\alpha\delta\over\beta^2}\, \big[(1+\alpha){\delta\over\beta}-1\big],
\\
&{\partial Z(\boldsymbol{\theta}_\delta)\over \partial \delta}
=
{2\alpha\over\beta}\, \big[(1+\alpha){\delta\over\beta}-1\big],
&
{\partial^2 Z(\boldsymbol{\theta}_\delta)\over \partial \alpha^2}
=
{2\delta^2\over\beta^2},
\\
&{\partial^2 Z(\boldsymbol{\theta}_\delta)\over \partial \beta^2}
=
{2\alpha\delta\over\beta^3}\, \big[(1+\alpha){3\delta\over\beta}-2\big],
&
{\partial^2 Z(\boldsymbol{\theta}_\delta)\over \partial \delta^2}
=
{2\alpha(1+\alpha)\over \beta^2};
\end{array}
\end{align}
and
\begin{align*}
&\textstyle
{\partial^2 Z(\boldsymbol{\theta}_\delta)\over \partial \alpha\partial\beta}
=
{\partial^2 Z(\boldsymbol{\theta}_\delta)\over \partial\beta\partial \alpha}
=
- \textstyle
{2\delta\over\beta^2}\,
\big[(1+2\alpha){\delta\over\beta}-1\big],
\\
&\textstyle
{\partial^2 Z(\boldsymbol{\theta}_\delta)\over \partial \alpha\partial\delta}
=
{\partial^2 Z(\boldsymbol{\theta}_\delta)\over \partial\delta\partial \alpha}
=
{2\over\beta}\, \big[(1+2\alpha){\delta\over\beta}-1\big],
\\
& \textstyle
{\partial^2 Z(\boldsymbol{\theta}_\delta)\over \partial \beta\partial\delta}
=
{\partial^2 Z(\boldsymbol{\theta}_\delta)\over \partial\delta\partial \beta}
=
-{2\alpha\over\beta^2}\, \big[(1+\alpha){2\delta\over\beta}-1\big].
\end{align*}
%
Note that $f(x;\boldsymbol{\theta}_\delta)$
is a positive, differentiable function of $\boldsymbol{\theta}_\delta=(\alpha,\beta,\delta)$. If a supremum
$\boldsymbol{\widehat{\theta}}$ exists, it must satisfy the likelihood equations
\begin{align}\label{likelihood equations}
{\partial l(\boldsymbol{\widehat{\theta}};{\bf x})\over\partial\alpha}=0,
\quad
{\partial l(\boldsymbol{\widehat{\theta}};{\bf x})\over\partial\beta}=0,
\quad
{\partial l(\boldsymbol{\widehat{\theta}};{\bf x})\over\partial\delta}=0.
\end{align}
Any (nontrivial) root of the likelihood equations \eqref{likelihood equations} is called an
ML estimator in the loose sense. In the case that the parameter value provides the absolute maximum
of $l(\boldsymbol{\theta}_\delta;{\bf x})$, it is called an ML estimator in the strict sense.

Also notice that, using the polygamma function of order $m$,
$\Psi^{(m)}(z)= {{\rm d}^{m+1}\over {\rm d}z^{m+1}}\log \Gamma(z)$,
the first-order partial derivatives of $l(\boldsymbol{\theta}_\delta;{\bf x})$ are
\begin{align}
\textstyle
{\partial l(\boldsymbol{\theta}_\delta;{\bf x})\over\partial\alpha}
&
\textstyle
=
-{1\over Z(\boldsymbol{\theta}_\delta)}\, {\partial Z(\boldsymbol{\theta}_\delta)\over \partial \alpha}
+
{\log \beta}
-
\Psi^{(0)}(\alpha)
+
\sum_{i=1}^{n}\log x_i    \label{1-der}
\\
&
\textstyle
=
- \frac{\delta [(1+2\alpha)\delta - 2\beta]}{2\beta^2 + \alpha\delta[(1+\alpha)\delta-2\beta]}
+\log{\beta} - \Psi^{(0)}(\alpha) + \sum^{n}_{i=1}\log{x_i} , \nonumber
\\
\textstyle
{\partial l(\boldsymbol{\theta}_\delta;{\bf x})\over\partial\beta}
&
\textstyle
=
-{1\over Z(\boldsymbol{\theta}_\delta)}\, {\partial Z(\boldsymbol{\theta}_\delta)\over \partial \beta}
+
{\alpha\over\beta}
-
n\overline{x} \nonumber
\\
&
\textstyle
=
\frac{2\alpha\delta[(1+\alpha)\delta-\beta]}
{2\beta^3 + \alpha\delta\beta[(1+\alpha) - 2\beta] } + \frac{\alpha}{\beta} - n\bar{x}, \nonumber
\\
\textstyle
{\partial l(\boldsymbol{\theta}_\delta;{\bf x})\over\partial\delta}
& \textstyle
=
-{1\over Z(\boldsymbol{\theta}_\delta)}\, {\partial Z(\boldsymbol{\theta}_\delta)\over \partial \delta}
-
2\sum_{i=1}^{n} {1-\delta{x}_i\over 1+(1-\delta{x}_i)^2} \nonumber
\\
& \textstyle
=
-\frac{ 2\alpha[(1+\alpha)\delta - \beta]}{2\beta^2 + \alpha\delta[(1+\alpha)\delta - 2\beta]} - 2 \sum^{n}_{i=1}\frac{1-\delta x_i}{1+(1-\delta x_i)^2 }. \nonumber
\end{align}
Since the equations in \eqref{likelihood equations} are not linear,
numerical methods will be used to solve the problem.  The solutions were found using Nelder-Mead method, since it is popular for unconstrained optimization and it is parsimonious in function evaluations per
iteration (\cite{lagarias1998}).

\color{black}
The second-order partial derivatives of $l(\boldsymbol{\theta}_\delta;{\bf x})$ can be written as
\begin{align}
    & \textstyle
    {\partial^2 l(\boldsymbol{\theta}_\delta;{\bf x})\over\partial\alpha^2}
    =
D_{\boldsymbol{\theta}_\delta}(\alpha,\alpha)
    -
    \Psi^{(1)}(\alpha), \label{2-der}
    \\
    & \textstyle
    {\partial^2 l(\boldsymbol{\theta}_\delta;{\bf x})\over\partial\beta^2}
    =
D_{\boldsymbol{\theta}_\delta}(\beta,\beta)
    -
    {\alpha\over\beta^2}, \nonumber
    \\
    & \textstyle
    {\partial^2 l(\boldsymbol{\theta}_\delta;{\bf x})\over\partial\delta^2}
    =
D_{\boldsymbol{\theta}_\delta}(\delta,\delta)
    +2
    \sum_{i=1}^{n} {x_i[1-(1-\delta{x}_i)^2]\over [1+(1-\delta{x}_i)^2]^2}; \nonumber
\end{align}
and the second-order mixed derivatives of $l(\boldsymbol{\theta}_\delta;{\bf x})$ are given by
\begin{align*}
& \textstyle
{\partial^2 l(\boldsymbol{\theta}_\delta;{\bf x})\over\partial\alpha\partial\beta}
=
{\partial^2 l(\boldsymbol{\theta}_\delta;{\bf x})\over\partial\beta\partial\alpha}
=
D_{\boldsymbol{\theta}_\delta}(\alpha,\beta)
+
{1\over\beta},
\\
& \textstyle
{\partial^2 l(\boldsymbol{\theta}_\delta;{\bf x})\over\partial\alpha\partial\delta}
=
{\partial^2 l(\boldsymbol{\theta}_\delta;{\bf x})\over\partial\delta\partial\alpha}
=
D_{\boldsymbol{\theta}_\delta}(\alpha,\delta)
,
\\
& \textstyle
{\partial^2 l(\boldsymbol{\theta}_\delta;{\bf x})\over\partial\beta\partial\delta}
=
{\partial^2 l(\boldsymbol{\theta}_\delta;{\bf x})\over\partial\delta\partial\beta}
=
D_{\boldsymbol{\theta}_\delta}(\beta,\delta)
,
\end{align*}
where
$
D_{\boldsymbol{\theta}_\delta}(u,v)
\coloneqq
{1\over Z(\boldsymbol{\theta}_\delta)}
\big[
{1\over Z(\boldsymbol{\theta}_\delta)}\,
{\partial Z(\boldsymbol{\theta}_\delta)\over \partial u}\,
{\partial Z(\boldsymbol{\theta}_\delta)\over \partial v}
-
{\partial^2 Z(\boldsymbol{\theta}_\delta)\over \partial u \partial v}
\big],
$
$u,v\in\{\alpha,\beta,\delta\}$.
Here, by the well-known Schwarz's Theorem,
the mixed partial differentiations are commutative at a given point $\boldsymbol{\theta}_\delta$
in $\mathbb{R}^3$ because the corresponding functions have continuous second partial
derivatives at that point.

If $X\sim \text{BGamma}(\boldsymbol{\theta}_\delta)$,
under mild regularity conditions the
Fisher information
matrix is given by
\[
I_X(\bm{\theta}_{\delta})
=
-
\begin{bmatrix}
D_{\boldsymbol{\theta}_\delta}(\alpha,\alpha)-\Psi^{(1)}(\alpha) & D_{\boldsymbol{\theta}_\delta}(\alpha,\beta)+{1\over\beta} & D_{\boldsymbol{\theta}_\delta}(\alpha,\delta)
\\
D_{\boldsymbol{\theta}_\delta}(\alpha,\beta)+{1\over\beta} & D_{\boldsymbol{\theta}_\delta}(\beta,\beta)-{\alpha\over\beta^2} & D_{\boldsymbol{\theta}_\delta}(\beta,\delta)
\\
D_{\boldsymbol{\theta}_\delta}(\alpha,\delta) & D_{\boldsymbol{\theta}_\delta}(\beta,\delta) & D_{\boldsymbol{\theta}_\delta}(\delta,\delta)
+2\,
\mathbb{E}{X [1-(1-\delta{X})^2]\over [1+(1-\delta{X})^2]^2}
\end{bmatrix},
\]
where
$\mathbb{E}\big|{X[1-(1-\delta{X})^2]\over [1+(1-\delta{X})^2]^2}\big|
\leqslant
2\,\mathbb{E}X-2\delta\,\mathbb{E}X^2
+
\delta^2\,\mathbb{E}X^3<\infty,
$
see Proposition \ref{prop-1}.
\begin{theorem}
Let $\Theta=\{\alpha\in\mathbb{R}^+: \varepsilon_0<\alpha<\alpha_0\}$ be the parameter space,
where $\varepsilon_0=\varepsilon_0(\beta,\delta)\in (0,\alpha_0)$ is fixed and
$\alpha_0=\alpha_0(\beta,\delta)\coloneqq [(2\beta-\delta)+\sqrt{(2\beta-\delta)^2+2\delta(4\beta-\delta)}\,]/2\delta$
with $\beta,\delta$ known such that $0<\delta<2\beta$.
Then, with probability approaching $1$, as $n\to\infty,$
the likelihood equation
${{\rm d}\, l(\alpha;{\bf x})\over{\rm d}\alpha}=0$ has a consistent solution,
denoted by $\widehat{\alpha}$.
\end{theorem}
\begin{proof}
Since $\beta$ and $\delta$ are known, to simplify the notation, we will write
$\text{BGamma}(\alpha)$, $f(x;\alpha)$, $Z(\alpha)$ and
$D_{\alpha}(\alpha,\alpha)$ refering to
$\text{BGamma}(\boldsymbol{\theta}_\delta)$, $f(x;\boldsymbol{\theta}_\delta)$,
$Z(\boldsymbol{\theta}_\delta)$ and
$D_{\boldsymbol{\theta}_\delta}(\alpha,\alpha)$, respectively.

Let $X\sim \text{BGamma}(\alpha)$.
By \cite{Cramer46} it is sufficient to prove that
\begin{enumerate}
\item $\mathbb{E}{{\rm d} \log f(X;\alpha)\over{\rm d}\alpha}=0$
for all $\alpha\in\Theta$;
\item $-\infty<\mathbb{E}{{\rm d}^2 \log f(X;\alpha)\over{\rm d}\alpha^2}<0$
for all $\alpha\in\Theta$;
\item There exits a function $H(x)$ such that for all $\alpha\in\Theta$,
\[
\biggl|{{\rm d}^3 \log f(x;\alpha)\over{\rm d}\alpha^3}\biggr|<H(x) \quad \text{and} \quad
\mathbb{E}H(X)=M(\alpha)<\infty.
\]
\end{enumerate}

Indeed,
taking $n=1$ in \eqref{1-der} we have
\begin{align*}
{{\rm d} \log f(x;\alpha) \over{\rm d}\alpha}
=
-{1\over Z(\alpha)}\, {{\rm d} Z(\alpha)\over {\rm d} \alpha}
+
{\log \beta}
-
\Psi^{(0)}(\alpha)
+
\log x.
\end{align*}
Then,
\begin{align*}
\mathbb{E}{{\rm d} \log f(X;\alpha) \over{\rm d}\alpha}
=
\mathbb{E}\log X
-
{1\over Z(\alpha)}\, {{\rm d} Z(\alpha)\over {\rm d} \alpha}
+
{\log \beta}
-
\Psi^{(0)}(\alpha).
\end{align*}
Using the Proposition \eqref{exp-log}-(1) and the identities in \eqref{der-1-2}, a straightforward
computation shows that
\[
\mathbb{E}\log X
-
{1\over Z(\alpha)}\, {{\rm d} Z(\alpha)\over {\rm d} \alpha}
=
\Psi^{(0)}(\alpha)-{\log \beta}.
\]
Therefore, $\mathbb{E}{{\rm d} \log f(X;\alpha)\over{\rm d}\alpha}=0$
for all $\alpha\in\Theta$, and the Item (1) is proved.

Taking $n=1$ in \eqref{2-der}, using the definition of $D_{\alpha}(\alpha,\alpha)$ and the identities in \eqref{der-1-2},  it follows that
\begin{align}\label{se-derivative}
{{\rm d}^2 \log f(x;\alpha) \over{\rm d}\alpha^2}
&=
D_{\alpha}(\alpha,\alpha)
-
\Psi^{(1)}(\alpha)
\\
&=
{\delta^3/\beta^4\over Z^2(\alpha)}\,
\big[
2\delta\alpha^2-2(2\beta-\delta)\alpha-(4\beta-\delta)
\big]
-
\Psi^{(1)}(\alpha). \nonumber
\end{align}
Since $0<\delta<2\beta$, $\alpha_0$ is well defined.
For $0<\alpha<\alpha_0$, note that
$2\delta\alpha^2-2(2\beta-\delta)\alpha-(4\beta-\delta)<0$.
On the other hand, its known that
$\Psi^{(1)}(\alpha)>{{\rm e}^{1/\alpha}\over ({\rm e}^{1/\alpha}-1)\alpha^2}>0$
(see \cite{article}, Corollary 1.2).
Therefore,
${{\rm d}^2 \log f(x;\alpha)\over{\rm d}\alpha^2}<0$ for all $x>0$. Hence,
the Item (2) is satisfied.

To prove Item (3), deriving with respect to $\alpha$ in \eqref{se-derivative} we obtain
\begin{align*}
{{\rm d}^3 \log f(x;\alpha) \over{\rm d}\alpha^3}
&=
{\delta^3/\beta^6\over Z^3(\alpha)}\,
\Big\{
-2\delta\big[(1+2\alpha)\delta-2\beta\big]
\big[2\delta\alpha^2-2(2\beta-\delta)\alpha-(4\beta-\delta)\big]
\\
& \quad +
\big[4\alpha\delta-2(2\beta-\delta)\big]
\big\{2\beta^2+\alpha\delta[(1+\alpha)\delta-2\beta]\big\}
\Big\}
-
\Psi^{(2)}(\alpha).
\end{align*}
Let
$
G(\alpha)\coloneqq
{\delta^3/\beta^6\over Z^3(\alpha)}\,
\big\{
2\delta\big[(1+2\alpha)\delta+2\beta\big]
\big[2\delta\alpha^2+2(2\beta+\delta)\alpha+(4\beta+\delta)\big]
+
\big[4\alpha\delta+2(2\beta+\delta)\big]
\{2\beta^2+\alpha\delta[(1+\alpha)\delta+2\beta]\}
\big\}.
$
Then, for all $x>0$ and $\alpha\in\Theta$,
\begin{align}\label{in-0}
\bigg|{{\rm d}^3 \log f(x;\alpha) \over{\rm d}\alpha^3}\biggr|
\leqslant G(\alpha)+|\Psi^{(2)}(\alpha)|.
\end{align}
Since $G(\alpha)$ is a increasing function in $\alpha$ and
$\alpha<\alpha_0$, we have
\begin{align}\label{in-1}
G(\alpha)\leqslant G(\alpha_0), \quad \text{for all} \ \alpha\in\Theta.
\end{align}

Combining the inequalities
$\Psi^{(n)}(\alpha)>-(n-1)!\, {\rm e}^{-n \Psi^{(0)}(\alpha)}$, for $n$ even
(see the inequality just below Item (2.9) from \cite{BATIR2007452}), and
$\Psi^{(0)}(\alpha)>\log(\alpha+{1\over 2})-{1\over \alpha}$ (see \cite{ELGIPE00}),
we have that
$-\Psi^{(2)}(\alpha)
<{\rm e}^{2[{1\over\alpha}-\log(\alpha+{1\over 2})]}
<{\rm e}^{2[{1\over\varepsilon_0}-\log(\varepsilon_0+{1\over 2})]},
$
for all $\alpha\in\Theta$.
On the other hand, by Corollary 1.2 from \cite{article},
$\Psi^{(2)}(\alpha)
<
{{\rm e}^{1/\alpha}[1-2\alpha({\rm e}^{1/\alpha}-1) ]
    \over ({\rm e}^{1/\alpha}-1)^2\alpha^4}<0$.
Therefore,
\begin{align}\label{in-2}
|\Psi^{(2)}(\alpha)|=-\Psi^{(2)}(\alpha)<{\rm e}^{2[{1\over\varepsilon_0}-\log(\varepsilon_0+{1\over 2})]}.
\end{align}
Combining \eqref{in-0}, \eqref{in-1} and \eqref{in-2},
\[
\bigg|{{\rm d}^3 \log f(x;\alpha) \over{\rm d}\alpha^3}\biggr|
< G(\alpha_0)+ {\rm e}^{2[{1\over\varepsilon_0}-\log(\varepsilon_0+{1\over 2})]},
\quad \text{for all} \ \alpha\in\Theta.
\]
Taking $H(x)=G(\alpha_0)+ {\rm e}^{2[{1\over\varepsilon_0}-\log(\varepsilon_0+{1\over 2})]}=
{\rm constante}$,
the proof of Item (3) follows. Thus, the proof of theorem is complete.
\end{proof}

\section{Monte Carlo simulation}\label{sec:05}

We here carry out a Monte Carlo simulation study to evaluate the performance of the
ML estimators of the BGamma model. All numerical evaluations were done in
the \texttt{R} software [\texttt{www.r-project.org}]. The simulation study considers the following scenario: sample
size $n \in \{10, 60, 120\}$, true shape parameter $\alpha \in \{0.50,1.00,1.50\}$,
true scale parameter $\beta \in \{1.00\}$, true value of the asymmetric parameter as $\delta \in \{-10,-5,-1,1,5,10\}$, with 5,000 Monte Carlo replications for each sample size.

For each value of the parameter $\delta$ and sample size,
the empirical values for the bias and mean squared error (MSE)
of the ML estimators are reported in Table~\ref{tab:simulatedbias}.  A look at the results in this table allows us to conclude that, as the sample size increases, the bias and MSE of all the estimators decrease, indicating that they are asymptotically unbiased, as expected.

\color{black}
\begin{table}[!ht]
 \footnotesize
\centering
\caption{Simulated values of biases (MSEs within parentheses) of the estimators of the BGamma model.}
\label{tab:simulatedbias}
\renewcommand{\arraystretch}{1.3}
\resizebox{\linewidth}{!}{
\begin{tabular}{lrrrrrrr}
\hline
 &  & \multicolumn{2}{l}{BG($\alpha = 0.50, \beta = 1.00, \delta$)} & \multicolumn{2}{l}{BG($\alpha = 1.00, \beta = 1.00, \delta$)} & \multicolumn{2}{l}{BG($\alpha = 1.50, \beta = 1.00, \delta$)} \\ \cline{3-8}
$n$ & $\delta$ & Bias($\widehat{\alpha}$) & Bias($\widehat{\beta}$) & Bias($\widehat{\alpha}$) & Bias($\widehat{\beta}$) & Bias($\widehat{\alpha}$) & Bias($\widehat{\beta}$) \\ \hline
20 & -10 & 0.4053 (0.7019) & 0.2028 (0.1850) & 0.5157 (1.3893) & 0.2018 (0.2165) & 0.5574 (2.2031) & 0.1770 (0.2427) \\[-1.2ex]
 & -5 & 0.2785 (0.3735) & 0.1617 (0.1328) & 0.4580 (1.0668) & 0.1895 (0.1878) & 0.5453 (1.8772) & 0.1795 (0.2199) \\[-1.2ex]
 & 1 & 0.0336 (0.0225) & 0.1196 (0.1426) & 0.1014 (0.0965) & 0.0810 (0.0592) & 0.2176 (0.3481) & 0.0931 (0.0685) \\[-1.2ex]
 & 5 & 0.3599 (0.6604) & 0.1661 (0.1435) & 0.6729 (1.8993) & 0.2345 (0.2417) & 0.7525 (2.7966) & 0.2253 (0.2561) \\[-1.2ex]
 & 10 & 0.5653 (1.2008) & 0.2435 (0.2370) & 0.4242 (2.0143) & 0.1161 (0.3131) & 0.6814 (2.7151) & 0.2084 (0.2587) \\[-1.2ex]
60 & -10 & 0.1294 (0.1080) & 0.0632 (0.0302) & 0.1567 (0.2546) & 0.0601 (0.0391) & 0.0321 (0.5791) & 0.0021 (0.0697) \\[-1.2ex]
 & -5 & 0.0774 (0.0503) & 0.0449 (0.0215) & 0.1392 (0.1843) & 0.0565 (0.0330) & 0.0835 (0.4332) & 0.0207 (0.0561) \\[-1.2ex]
 & 1 & -0.0033 (0.0059) & 0.0134 (0.0191) & 0.0287 (0.0241) & 0.0225 (0.0142) & 0.0623 (0.0699) & 0.0268 (0.0152) \\[-1.2ex]
 & 5 & 0.0766 (0.0516) & 0.0373 (0.0172) & 0.2218 (0.3057) & 0.0753 (0.0391) & 0.2363 (0.5169) & 0.0692 (0.0459) \\[-1.2ex]
 & 10 & 0.1912 (0.1766) & 0.0798 (0.0358) & -0.2430 (0.5223) & -0.1394 (0.1098) & 0.1682 (0.5757) & 0.0482 (0.0551) \\[-1.2ex]
120 & -10 & 0.0626 (0.0458) & 0.0286 (0.0124) & 0.0767 (0.1162) & 0.0275 (0.0169) & -0.1228 (0.3672) & -0.0502 (0.0448) \\[-1.2ex]
 & -5 & 0.0327 (0.0206) & 0.0173 (0.0089) & 0.0677 (0.0821) & 0.0257 (0.0141) & -0.0376 (0.2493) & -0.0217 (0.0328) \\[-1.2ex]
 & 1 & -0.0107 (0.0031) & -0.0054 (0.0085) & 0.0133 (0.0118) & 0.0090 (0.0065) & 0.0298 (0.0321) & 0.0114 (0.0068) \\[-1.2ex]
 & 5 & 0.0273 (0.0169) & 0.0118 (0.0068) & 0.1101 (0.1292) & 0.0354 (0.0159) & 0.1176 (0.2350) & 0.0326 (0.0197) \\[-1.2ex]
 & 10 & 0.0944 (0.0706) & 0.0373 (0.0140) & -0.4445 (0.3810) & -0.2162 (0.0848) & 0.0386 (0.3181) & 0.0073 (0.0300) \\ \hline
\end{tabular}}
\end{table}

\section{The $\text{BGamma}(\tetn_\delta)$ regression model with censored data}\label{sec:06}

In many practical applications, the lifetimes are affected by
explanatory variables such as sex, age, grade of disease,
tumor thickness and several others. So, it is important
to explore the relationship between the response variable and the explanatory variables. Regression models can be proposed in different forms in statistical analysis.

In this section, we define a parametric regression model using the new
distribution with censored data, called the $\text{BGamma}(\tetn_\delta)$ regression
model, for reliability analysis as a fea\-sible alternative to
the location-scale regression model. Considering that the $\text{BGamma}(\tetn_\delta)$ and
$\text{BGamma}(\tetn_0)$ regression models are embedded models, the LR statistic can be used to
discriminate between these models. We adopt a classic frequentist analysis
for the $\text{BGamma}(\bf{\theta}_\delta)$ regression model.

Regression analysis of lifetimes involves specifications for the
lifetime distribution of $X$ given a vector of covariates
denoted by $\vn=(v_1,\cdots,x_p)^{T}$. Here, we relate the parameters $\alpha$ and $\beta$ to
covariates by the logarithm link functions
$\alpha_{i}=\exp(\vn_{i}^{T}\taun_1)$ and
$\beta_{i}=\exp(\vn_{i}^{T}\taun_2)$, $i=1,\ldots,n,$ respectively,
where $\taun_1=(\tau_{11},\cdots,\tau_{1p})^{T}$ and
$\taun_2=(\tau_{21},\cdots,\tau_{2p})^{T}$ denote the vectors of
regression coefficients and $\vn_{i}^{T}=(v_{i1},\cdots,v_{ip})$.

The survival function of $X|\vn$ follows from (Proposition \ref{prop-3}) as
\begin{eqnarray}\label{reg_gollfw}
S(x|\vn)&=&1+\frac{\delta\,x^{\exp(\vn^{T}\taun_1)}\exp(\vn^{T}\taun_2)^{\exp(\vn^{T}\taun_1)-1}\exp[-\exp(\vn^{T}\taun_2)\,x]}
{\Gamma[\exp(\vn^{T}\taun_1)]
Z(\tetn_{\delta})}\times\nonumber\\&&\left\{\delta\left[x+\frac{\exp(\vn^{T}\taun_1)-1}{\exp(\vn^{T}\taun_2)}\right]-2\right\}
-I(\exp(\vn^{T}\taun_1),\exp(\vn^{T}\taun_2)\,x),
\end{eqnarray}
where
where $I(k,y)=\gamma(k,y)/\Gamma(k)$ is the incomplete gamma
ratio function, $\gamma(k,y)=\int_{0}^{y}w^{k-1}e^{-w}dw$ is the incomplete gamma function and $\Gamma(\cdot)$ is the gamma function.
Equation (\ref{reg_gollfw}) is referred to as the survival function for the $\text{BGamma}(\tetn_\delta)$ regression model, which
opens new possibilities for fitting many different types of reliability data.

Consider a sample $(x_1,\vn_1),\cdots,(x_n,\vn_n)$ of $n$
independent observations. We consider that each individual $i$ has a
lifetime $X_i$ and a censoring time $C_i$, where $X_i$ and $C_i$ are independent random variables
and the data consist of $n$ independent observations and $x_i =\min(X_i,C_i)$, for $i = 1,\ldots, n$.
We assume non-informative censoring such that the observed lifetimes and censoring times are
independent. Let $F$ and $C$ be the sets of individuals for which $x_i$ is the lifetime or censoring, respectively.
 Conventional likelihood estimation
techniques can be applied here. The total log-likelihood function for the
vector of parameters
$\psin=(\delta,\taun_{1}^{T},\taun_{2}^{T})^{T}$ from model
(\ref{reg_gollfw}) has the form
\begin{eqnarray}\label{veroy}
l(\psin)&=&\sum\limits_{i \in F}\log[1+(1-\delta\,x_{i})^{2}]\sum\limits_{i \in F}\log[Z_{i}(\tetn_{\delta})]
\sum\limits_{i \in F}\exp(\vn_{i}^{T}\taun_1)(\vn_{i}^{T}\taun_2)\nonumber \\&&-\sum\limits_{i \in F}\log\{\Gamma[\exp(\vn_{i}^{T}\taun_1)]\}
+\sum\limits_{i \in F}(\exp(\vn_{i}^{T}\taun_1)-1)\log(x_{i})-\sum\limits_{i \in F}\exp(\vn_{i}^{T}\taun_2)x_{i}
\nonumber\\&&+\sum\limits_{i \in C}l_{i}^{(c)}(\psin),
\end{eqnarray}
where
\begin{eqnarray*}
l_{i}^{(c)}(\psin)&=&\log\bigg\{1+\frac{\delta\,x_{i}^{\exp(\vn_{i}^{T}\taun_1)}\exp(\vn_{i}^{T}\taun_2)^{\exp(\vn_{i}^{T}\taun_1)-1}\exp[-\exp(\vn_{i}^{T}\taun_2)\,x_{i}]}
{\Gamma[\exp(\vn_{i}^{T}\taun_1)]
Z_{i}(\tetn_{\delta})}\times\nonumber\\&&\left\{\delta\left[x+\frac{\exp(\vn_{i}^{T}\taun_1)-1}{\exp(\vn_{i}^{T}\taun_2)}\right]-2\right\}
-I(\exp(\vn_{i}^{T}\taun_1),\exp(\vn_{i}^{T}\taun_2)\,x_{i})\bigg\}
\end{eqnarray*}
and
\begin{eqnarray*}
Z_{i}(\tetn_{\delta})=2+\frac{\exp(\vn_{i}^{T}\taun_1)\,\delta}{\exp(\vn_{i}^{T}\taun_2)}\left\{\left[1+\exp(\vn_{i}^{T}\taun_1)\right]
\left[\frac{\delta}{\exp(\vn_{i}^{T}\taun_2)}\right]-2\right\}.
\end{eqnarray*}

The MLE $\widehat{\psin}$
of the vector of unknown parameters can be determined by maximizing the log-likelihood (\ref{veroy}). We
use the $\texttt{R}$ software to compute $\widehat{\psin}$. Initial values
for $\taun_1$ and $\taun_2$ are taken from the fit of the $\text{BGamma}(\tetn_0)$
regression model with $\delta=0$.

The multivariate normal $N_{2p+1}(0,J(\widehat\psin)^{-1})$ distribution under standard
regularity conditions can be used to construct approximate confidence intervals for the model
parameters. Further, we can compare the $\text{BGamma}(\tetn_\delta)$ model with its special models using
LR statistics.

\section{Applications}\label{sec:07}

In this section, we provide two applications to real data to illustrate the flexibility of the $\text{BGamma}(\tetn_\delta)$ model. In the first application, we present a real situation in which the behavior of the data is bimodal. In the second application, we
consider a $\text{BGamma}(\tetn_\delta)$ regression model with censored data. In the applications, we determine the MLEs and the corresponding standard errors (SEs) (given
in parentheses) of the model parameters and the values of the Akaike Information Criterion (AIC), Bayesian Information Criterion (BIC), Cramer-von Mises ($W^{*}$) and Kolmogorov-Smirnov ($KS$) goodness-of-fit statistic for the fitted models.
For all cases, the model parameters are estimated by the ML method using the \texttt{R} software.

\subsection{Application 1: Wheaton River data}

The data are the exceedances of flood peaks (in $m^{3}/s$) of the Wheaton River near Carcross in Yukon Territory, Canada. The data consist of 72 exceedances for the years 1958--1984, rounded
to one decimal place. These data are presented and analyzed by \cite{chouste:01} and \cite{akifamlee:08}. In \cite{akifamlee:08} the authors present an analysis considering the following distributions: Pareto, three-parameter Weibull, generalized Pareto and Beta-Pareto. The authors use the KS measurement to select the most appropriate model. In Table \ref{beta:pareto} we present these values and the associated p-value.

\begin{table}[htb]
	\centering {\caption{The $\rm{KS}$ measurements and associated $p$-value with the Wheaton River data.\label{beta:pareto}}}\vspace*{0.15cm}
	\begin{tabular}{ccc}
		\hline
		Model    &    KS    & $p$-value  \\
		\hline
		Pareto   &    2.7029  &  $<$0.000\\
		Three-parameter Weibull     &  1.6734 &  0.0074  \\
		Generalized Pareto    &      1.205 &   0.1094 \\
		Beta Pareto    &     1.2534 & 0.0864   \\
		\hline
	\end{tabular}
\end{table}

We consider the Kumaraswamy generalized gamma (KumGG)  distribution (for $x>0$) defined by \cite{paorcor:11}.
Note that the KumGG distribution contains as particular cases most of the classical distributions used in survival analysis. Hence, the associated
density function with five positive parameters $\alpha$, $\tau$, $k$, $\lambda$
and $\varphi$ has the form
\begin{eqnarray*}\label{density}
f(x)&=&\frac{\lambda\,\varphi\,\tau}{\alpha\Gamma(k)}\left(\frac{x}{\alpha}\right)^{\tau k-1}\exp\Biggl[-\left(\frac{x}{\alpha}\right)^{\tau}\Biggl]
\Biggl\{\gamma_{1}\Biggl[k,\left(\frac{x}{\alpha}\right)^{\tau}\Biggl]\Biggl\}^{\lambda-1}\times\\&&
\Biggl(1-\Biggl\{\gamma_{1}\Biggl[k,\left(\frac{x}{\alpha}\right)^{\tau}\Biggl]\Biggl\}^{\lambda}\Biggl)^{\varphi-1},
\end{eqnarray*}
where $\gamma_{1}(k,y)=\gamma(k,y)/\Gamma(k)$ is the incomplete gamma
ratio function, $\alpha$ is a scale parameter and the other positive parameters $\tau$,
$k$, $\varphi$ and $\lambda$ are shape parameters. This model has as particular cases, exponentiated Weibull (for $\lambda=1$ and $\varphi=1$), gamma for (for $\lambda=1$, $\varphi=1$  and $\tau=1$) and Weibull for (for $\lambda=1$, $\varphi=1$  and $k=1$).

The modified Weibull (MW) (for $x\geq0$) was defined by \cite{laixiemur:03}, whose density function with three parameters $\alpha>0$, $\tau\geq0$ and $k\geq0$ is given by
\begin{eqnarray*}
	f(x) = \alpha\,x^{(\tau-1)} (\tau+k\,x)\exp[k\,x-\alpha\,x^{\tau}\exp(k\,x)].
\end{eqnarray*}

The results are reported in Tables \ref{EMV:aplica1} and \ref{EMV:aplica1_aic}. The four statistics agree on the model's ranking. The lowest values of these criteria correspond to the
$\text{BGamma}(\tetn_\delta)$ distribution, which could be chosen in this case. Also in relation to Table \ref{beta:pareto} we verified that the $KS$ measurement of the proposed model presents smaller values and associated $p$-value is higher, indicating that the model is adequate to the data of Wheaton River data.

\begin{table}[htb!]
	{\small
		\centering {\caption{MLEs of the model parameters for the  Wheaton River data .}
			\label{EMV:aplica1}}
		\vspace*{0.15cm}
		\begin{tabular}{c|ccccc}
			\hline
			Model  & $\alpha$ & $\beta$   &$\delta$ &  &     \\
			\hline
			$\text{BGamma}(\tetn_\delta)$ &  1.054  & 0.176 &  0.177  &       &                     \\
			& (0.145) & (0.111) & (0.032) &       &                   \\
			\hline
		  & $\alpha$ & $\tau$   &$k$ & $\lambda$ &    $\varphi$ \\
		\hline
		Kw-GG  &  548.542 & 0.103  &  0.098  & 158.570 & 869.87     \\
			& (252.1) & (0.082)& (0.007) & (70.450)        & (196.8)     \\
			gamma &  14.558  &1 &  0.838  &   1     &             1        \\
			& (2.816) & (-) & (0.121) &     (-)     &     (-)                     \\
				EW    &  11.278 & 1.380  &0.591    &     1        &    1       \\
			& (1.506)&  (0.284)& (0.149)  &       (-)          &    (-)        \\
			Weibull& 11.632 & 0.901   &    1      &      1      &    1      \\
			& (1.601) & (0.085) &   (-)    &       (-)            &      (-)      \\
			MW  & 0.124  & 0.775  &  0.010  &              &             \\
			& (0.034) & (0.124) & (0.007) &                  &            \\
					\hline
	\end{tabular}}
\end{table}

\begin{table}[htb!]
	{\small
		\centering {\caption{Statistical measures.}
			\label{EMV:aplica1_aic}}
		\vspace*{0.15cm}
		\begin{tabular}{c|cccc}
			\hline
			Model  &   $\rm{AIC}$ &  $\rm{BIC}$& $W^{*}$ & $KS$\\
			\hline
			$\text{BGamma}(\tetn_\delta)$ &   \textbf{ 501.51}  &  \textbf{ 508.34}& \textbf{0.038} & \textbf{0.065}\\
			&            &            &        &\textbf{(0.918)}\\
			\hline
		  &  $\rm{AIC}$ &  $\rm{BIC}$& $W^{*}$ & $KS$\\
		\hline
		Kw-GG  &     514.01  &  525.39 & 0.159 & 0.099\\
			&          &       &         &(0.473)\\
			gamma &            506.68    &    511.24 & 0.130  & 0.102\\
			&           &    & &(0.433)\\
				EW    &     505.85 &  512.68& 0.074& 0.096\\
			&        &       & &(0.516)\\
			Weibull& 506.99   &   511.55 & 0.137 & 0.105\\
			&          &         & &(0.402)\\
			MW  &    507.34 & 514.17& 0.097 &0.100\\
			&           &        & & (0.466) \\
					\hline
	\end{tabular}}
\end{table}

In Figure \ref{ajuste_aplica_river}, we present the adjustment of the proposed model in relation to the PDF and CDF; see Figures \ref{ajuste_aplica_river}(a,b). In Figure \ref{ajuste_aplica_river}(c), we provide the QQ plot for the $\text{BGamma}(\tetn_\delta)$ distribution. We note that the quantile residuals follow more approximately a normal distribution for the $\text{BGamma}(\tetn_\delta)$ distribution. In fact, these  plots reveal that the $\text{BGamma}(\tetn_\delta)$ distribution provides a good fit for  Wheaton River data.

\begin{figure}[htb!]
	\begin{center}
		(a)\hspace{5cm} (b)\hspace{5cm} (c)\\
		\includegraphics[width=5cm,height=6.5cm]{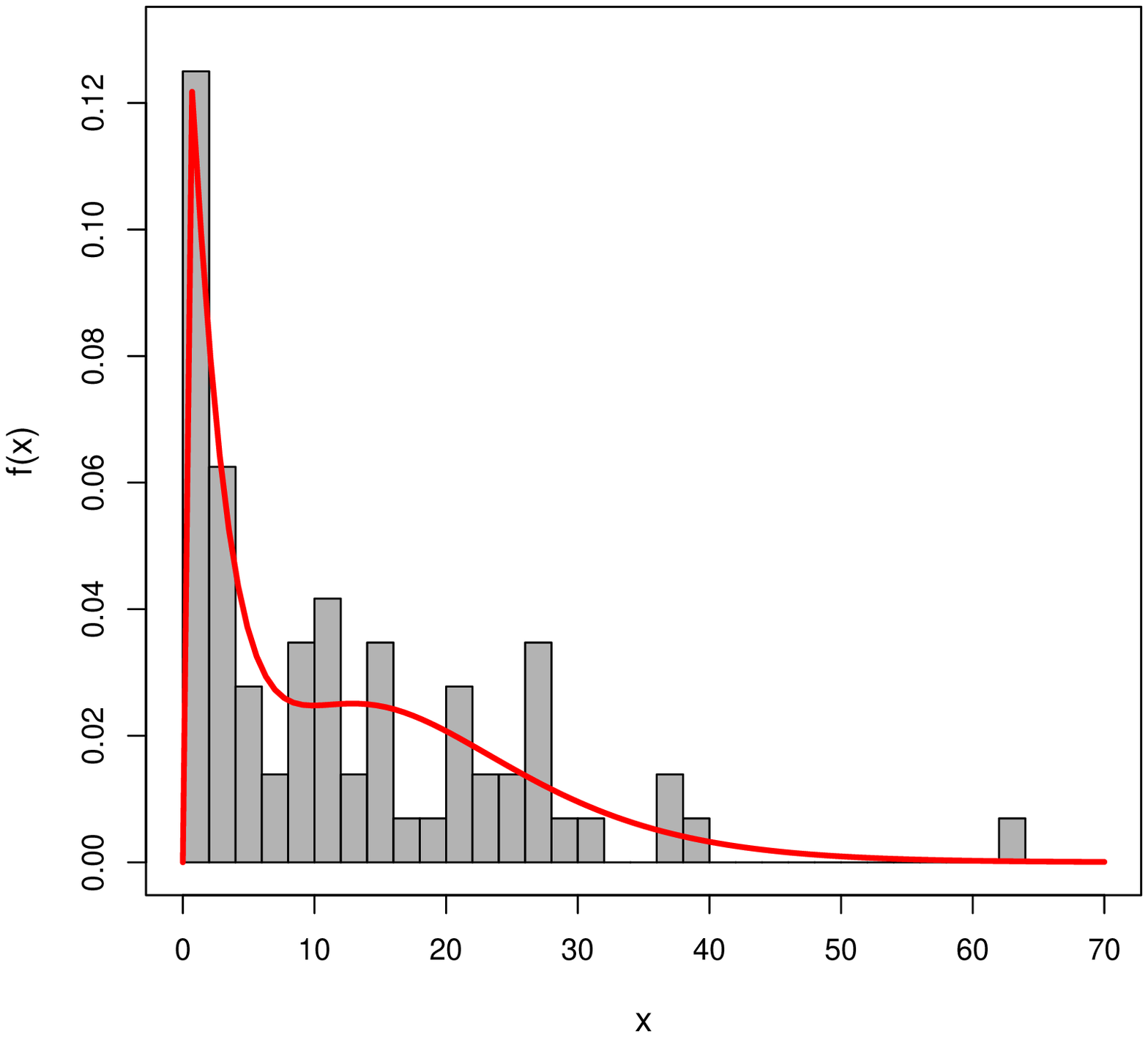}~
		~\includegraphics[width=5cm,height=6.5cm]{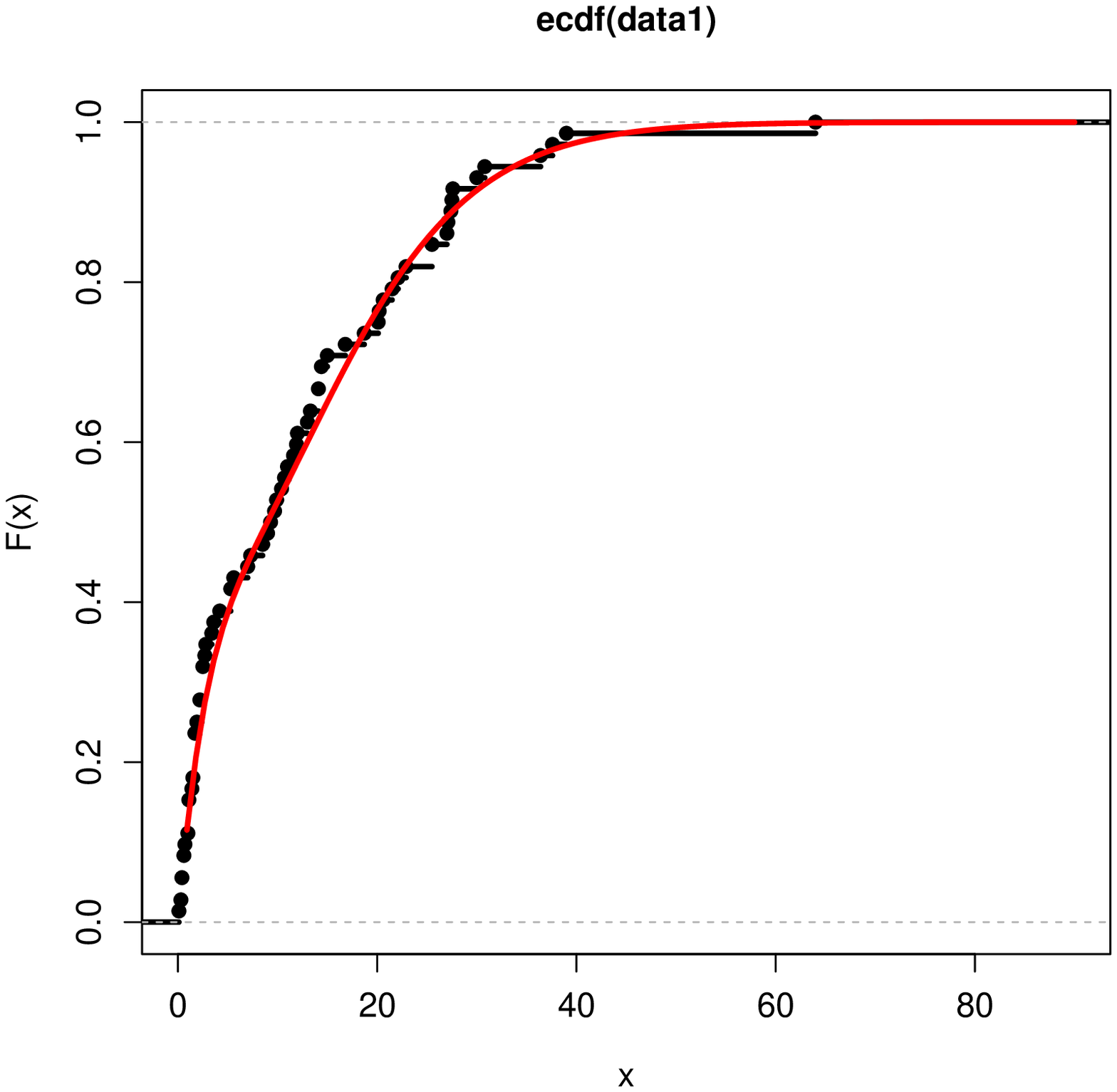}~
		~\includegraphics[width=5cm,height=6.5cm]{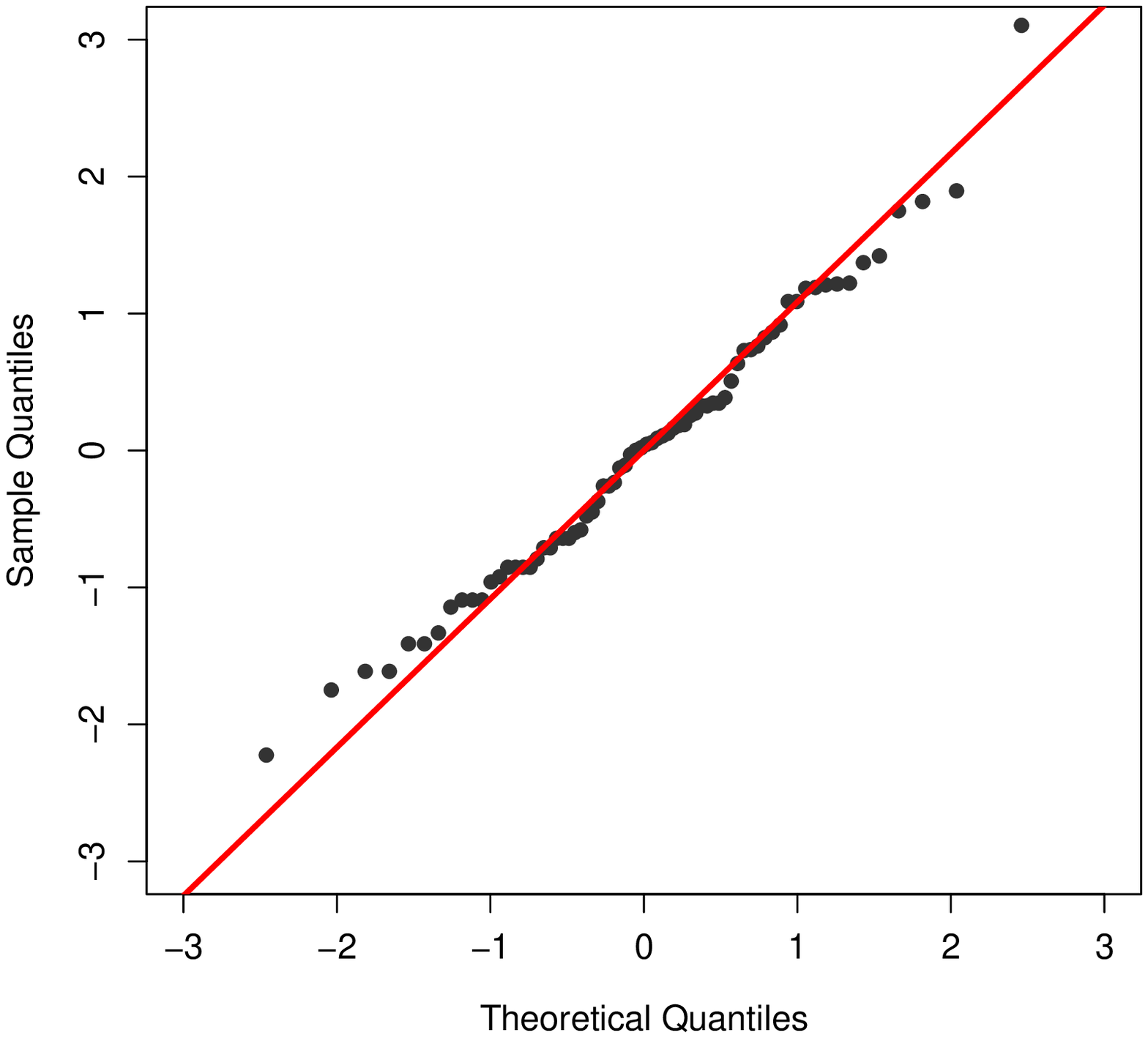}
		\caption{(a) Estimated PDF of the $\text{BGamma}(\tetn_\delta)$ model. (b) Empirical CDF and estimated CDF of the $\text{BGamma}(\tetn_\delta)$ model. (c) QQ plot  for the quantile residual  from the
fitted $\text{BGamma}(\tetn_\delta)$ model to the Wheaton River data.\label{ajuste_aplica_river}}
	\end{center}
\end{figure}

\subsection{Application 2: Gastric cancer data}
Stomach cancer is also known as gastric cancer. Stomach cancer develops slowly over many years. Prior to the appearance of the cancer itself, precancerous changes occur in the inner lining of the stomach (mucosa). These early changes rarely cause symptoms and therefore often go unnoticed.
Thus, new technologies to optimize medical decisions and the development of new therapies are of great importance to improve survival in gastric cancer.
In this second application, in order to illustrate the use of $\text{BGamma}(\tetn_\delta)$ regression, we consider the data set analyzed by \cite{martinezetal:13} and \cite{ortegaetal:17}. These last two surveys use the healing fraction regression model to analyze this gastric cancer data.
The sample size is $n=201$ patients of different clinical stages, of which 76 patients who received adjuvant chemoradiotherapy and 125 who received resection alone. The response variable refers to times to death in months since surgery. We observed that we have 53.2\% of the censored data.
Thus the variables used were:
\begin{itemize}
  \item $x_i$: time to death in months since surgery;
  \item $v_{i1}$: type of therapy (0=adjuvante chemoradiotherapy; 1=surgery alone) for $i=1,\ldots,271$.
\end{itemize}

We now present results by fitting the $\text{BGamma}(\tetn_\delta)$  regression model
\begin{eqnarray*}
\alpha_{i}=\exp(\tau_{10}+v_{i1}\tau_{11}) \qquad \mbox{and} \qquad \beta_{i}=\exp(\tau_{20}+v_{i1}\tau_{21}), \qquad  i=1,\ldots,271.
\end{eqnarray*}

The results in Table \ref{modelos:completos:2} indicate that the $\text{BGamma}(\tetn_\delta)$ regression model has the lowest GD  and AIC values among those of the fitted models, and so it could be chosen as the best regression model. If we consider the BIC statistic, then the $\text{BGamma}(\tetn_\delta)$ and gamma regressions models are more appropriate to model this data set.

\begin{table}[htb]
		\centering {\caption{The $\rm{GD}$, $\rm{AIC}$ and $\rm{BIC}$ measurements
				for the BG, gamma and Weibull regression models for the gastric cancer.\label{modelos:completos:2}}}\vspace*{0.3cm}
		\begin{tabular}{ccccc}
			\hline
			Model    &     GD      &    AIC         &   BIC         & \\
			\hline
			$\text{BGamma}(\tetn_\delta)$    &    \textbf{866.53}  &   \textbf{ 876.53}     &  \textbf{ 893.04 }    &   \\
			$\text{BGamma}(\tetn_0)$     &    871.28  &   879.28     &   892.49     &   \\
			Weibull    &     872.34  &    880.34     &  893.55     &     \\
			\hline
		\end{tabular}
	\end{table}
	
We note from the fitted  $\text{BGamma}(\tetn_\delta)$ regression model that $x_{1}$ is
significant (at $5\%$ level). Further, there is a significant
difference between type of therapy (adjuvante chemoradiotherapy and surgery alone)
for the time to death in months since surgery.
	\begin{table}[htb!]
		\centering {\caption{
				MLEs, SE and $p$-value for the parameters from the $\text{BGamma}(\tetn_\delta)$ regression model on the gastric cancer.}
			\label{estimativas}}
		\vspace*{0.3cm}
		\begin{tabular}{cccc}
			\hline
			Parameter       &   Estimate &    SE   &    $p$-Value    \\
			\hline
			$\tau_{10}$     &    -0.221 &  0.073   & 0.003\\
			$\tau_{11}$     &    1.306 &  0.100   &$<$0.001\\
			$\tau_{20}$     &   -3.506  &   0.077   & $<$0.001\\
			$\tau_{21}$     &   1.077 &  0.102   &$<$0.001\\
			$\delta$     &    0.032 &  0.002   &  \\
			\hline
		\end{tabular}
	\end{table}

In order to detect possible outlying observations as well as departures from the assumptions of $\text{BGamma}(\tetn_\delta)$ regression model, we present, in Figure \ref{residuos:quantile}, the plots of the density, QQ-plot and worm plot for the quantile residuals.	By analyzing these plots, we conclude that the $\text{BGamma}(\tetn_\delta)$ regression model provides a good adjustment.
	\begin{figure}[htb!]
		\begin{center}
			(a)\hspace{5cm} (b)\\
			\includegraphics[width=6cm,height=4.5cm]{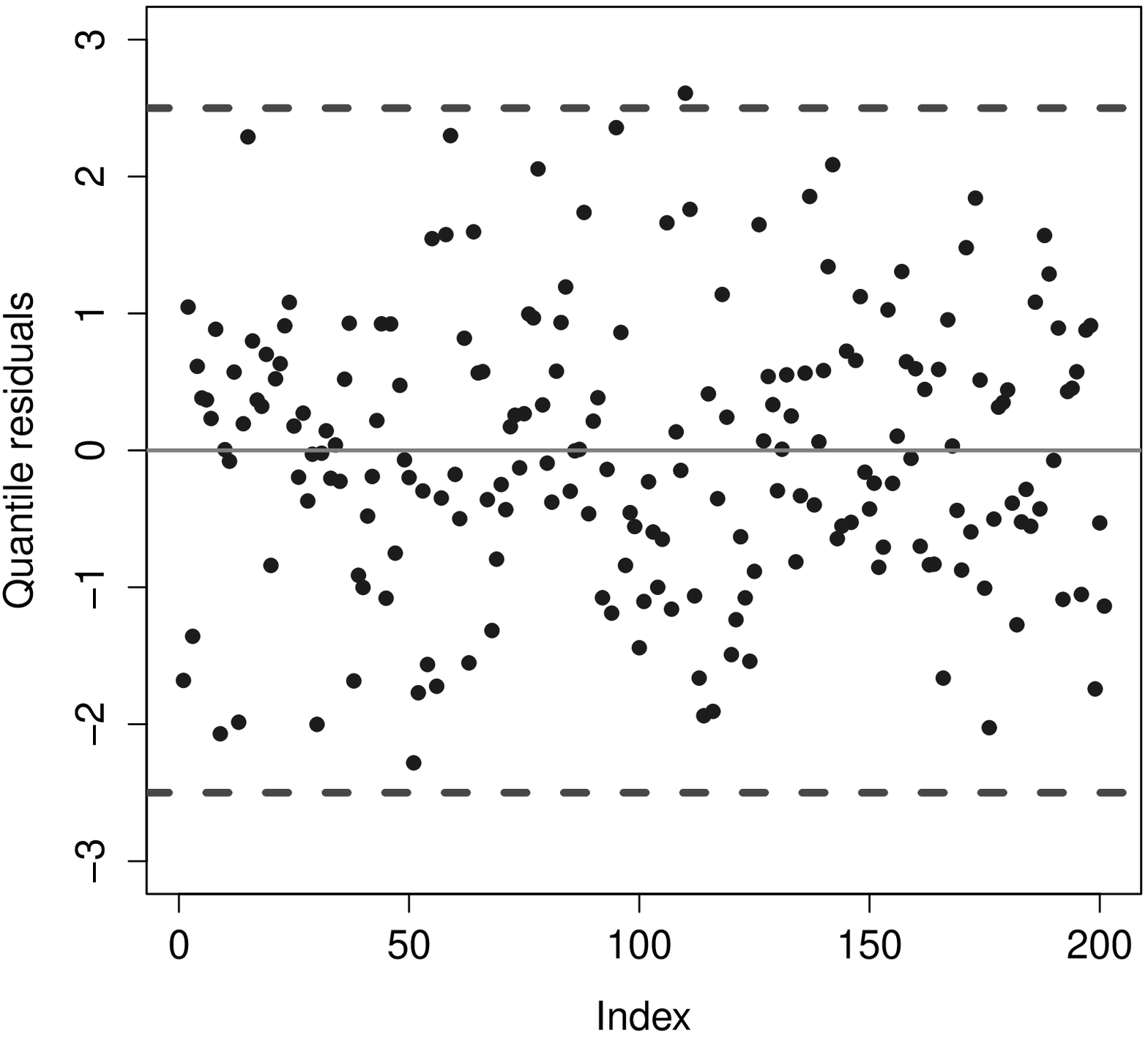}~
			~\includegraphics[width=6cm,height=4.5cm]{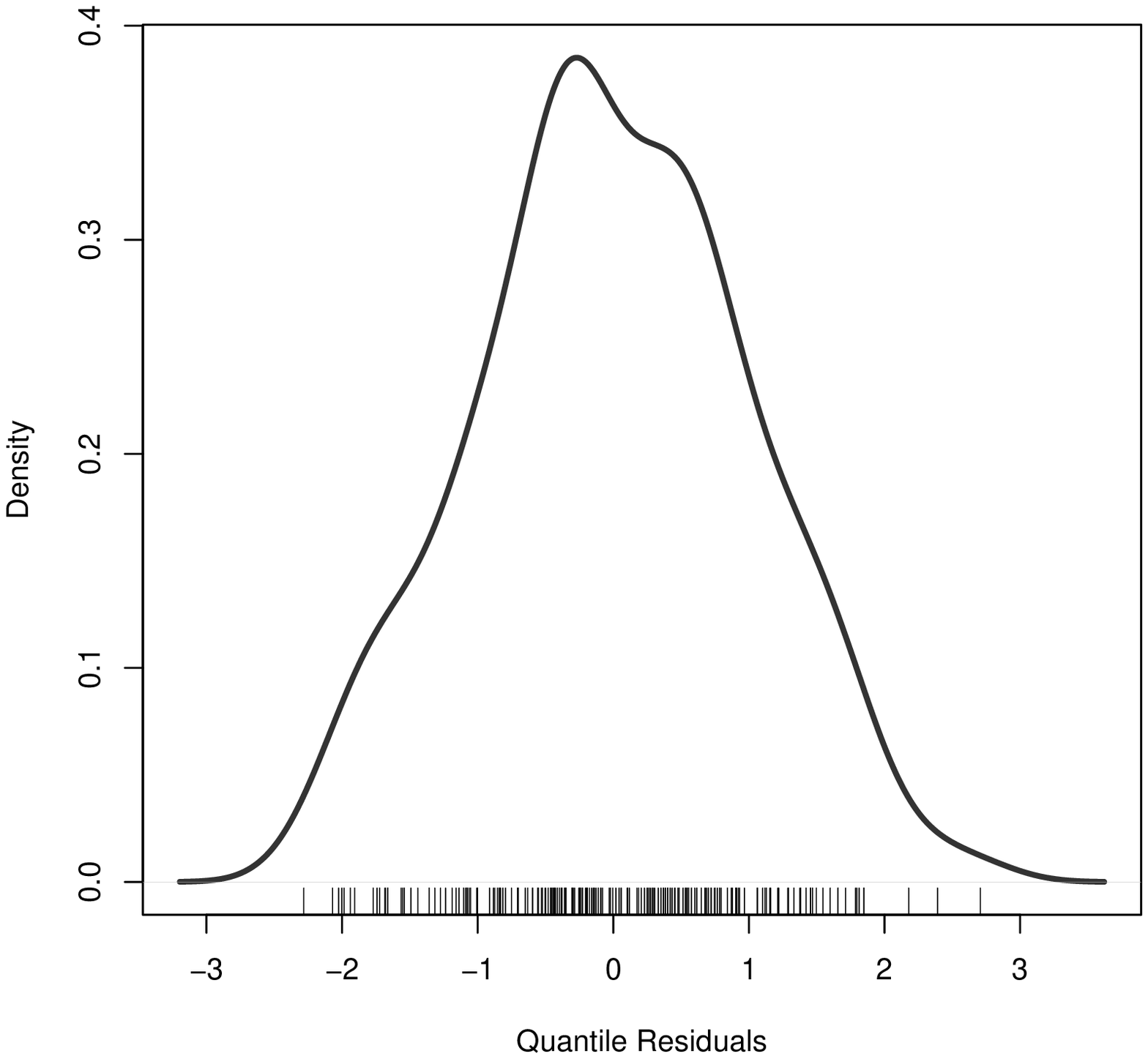}\\
			(c)\hspace{5cm} (d)\\
			\includegraphics[width=6cm,height=4.5cm]{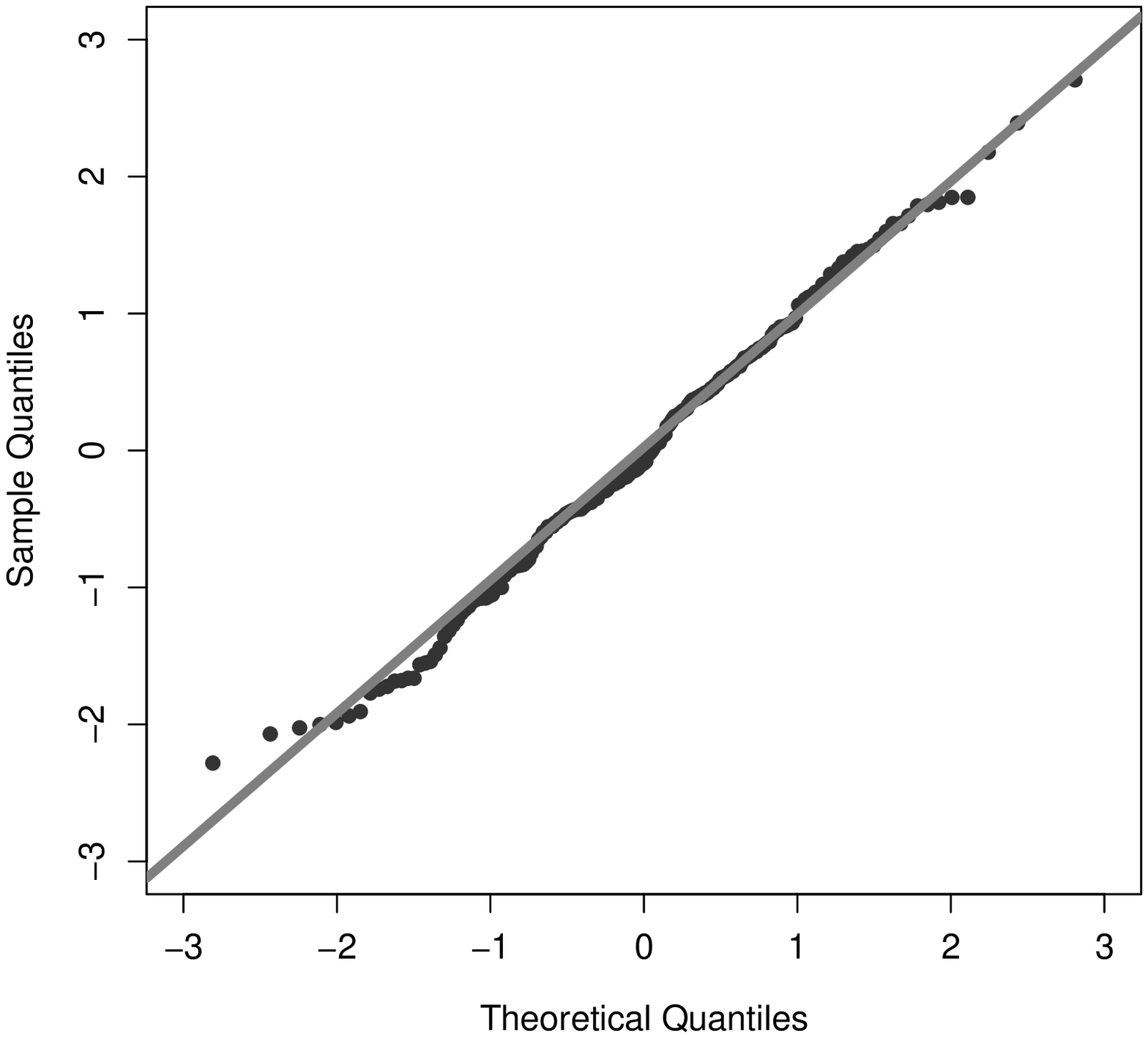}~
			~\includegraphics[width=6cm,height=4.5cm]{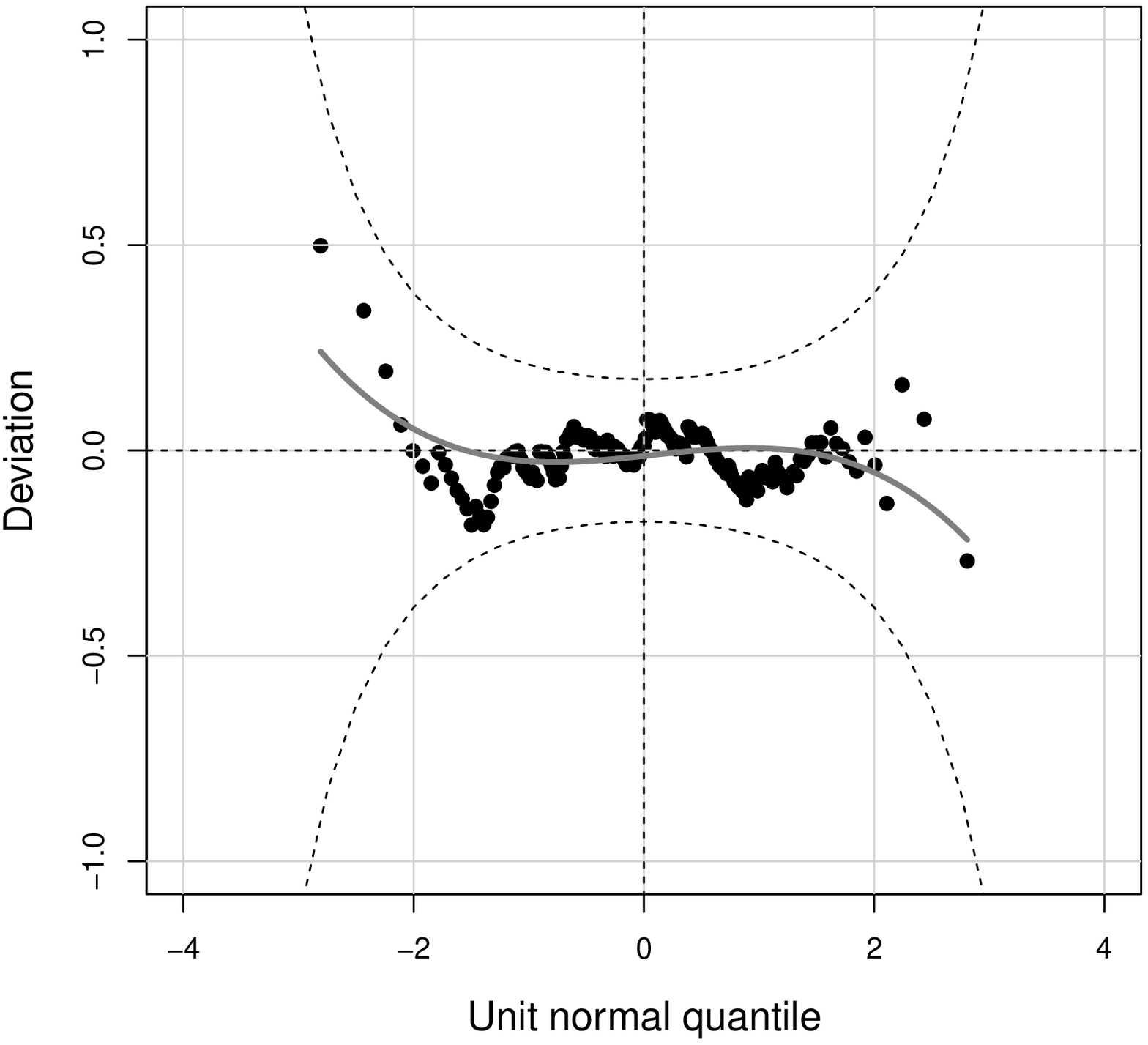}
			\caption{Plots of the diagnostic on fitting the BG regression model for gastric cancer. (a) Index plot for $\widehat{q}_i$  (b) Estimated density function for $\widehat{q}_i$. (c) QQ plot for $\widehat{q}_i$. (d) Worm plot for $\widehat{q}_i$.\label{residuos:quantile}}
		\end{center}
	\end{figure}
	
Finally, in order to assess if the model is appropriate, the empirical and estimated survival functions
of the $\text{BGamma}(\tetn_\delta)$ regression model are plotted in Figure \ref{dens_aplica_temp} for the different treatments.
Figure \ref{dens_aplica_temp}(a) shows the fit of the $\text{BGamma}(\tetn_\delta)$ regression model considering regression structure only in the $\alpha$ parameter. Figure \ref{dens_aplica_temp}(b) shows the fit considering two regression structures in the $\alpha$ and $\beta$ parameters.
We may conclude from
the plots that the $\text{BGamma}(\tetn_\delta)$ regression model considering two regression structures provides a suitable fit to the gastric cancer data.
		
	\begin{figure}[H]
		\begin{center}
			(a)\hspace{5cm} (b)\\
			\includegraphics[width=6cm,height=5.5cm]{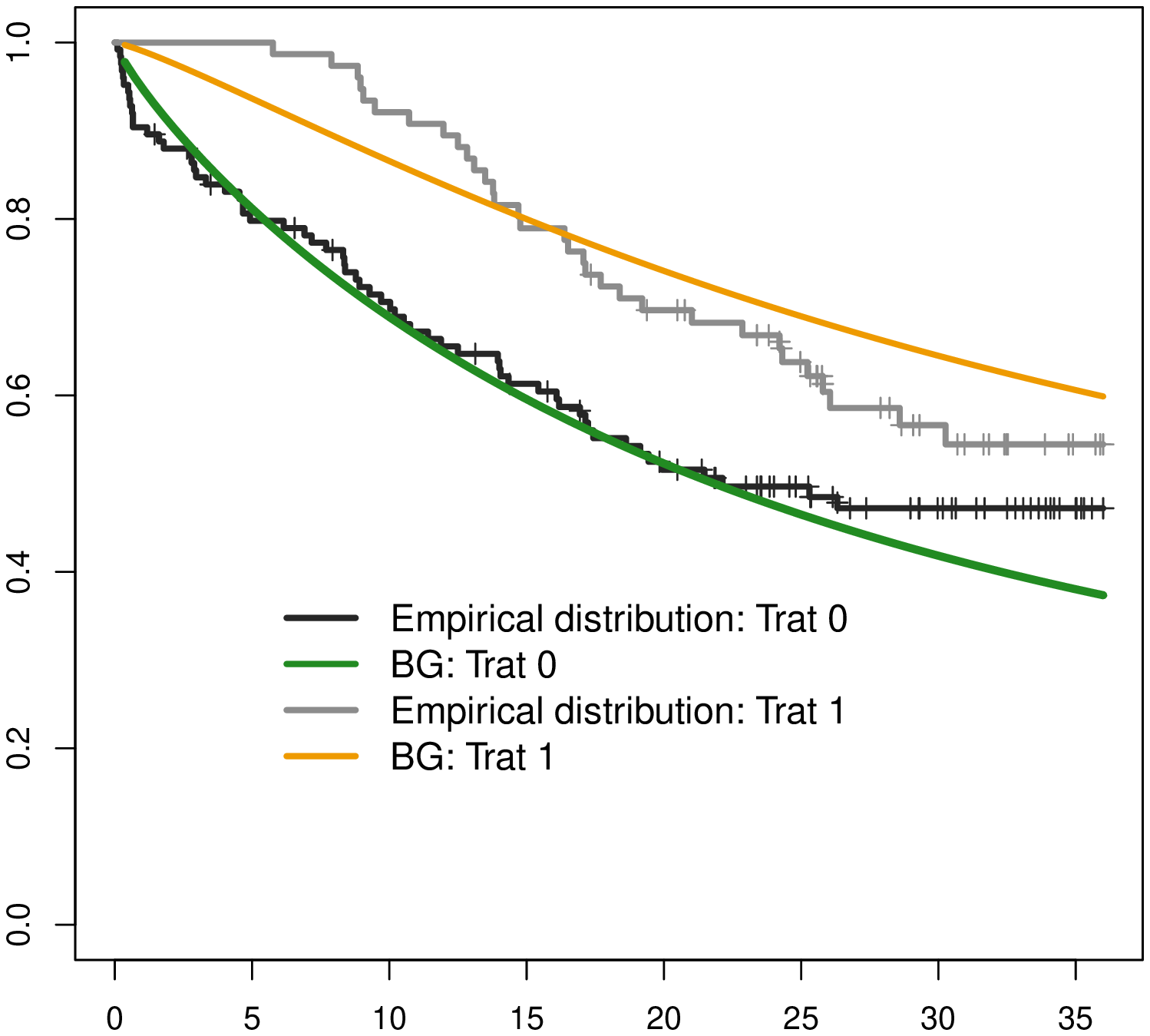}~
			~\includegraphics[width=6cm,height=5.5cm]{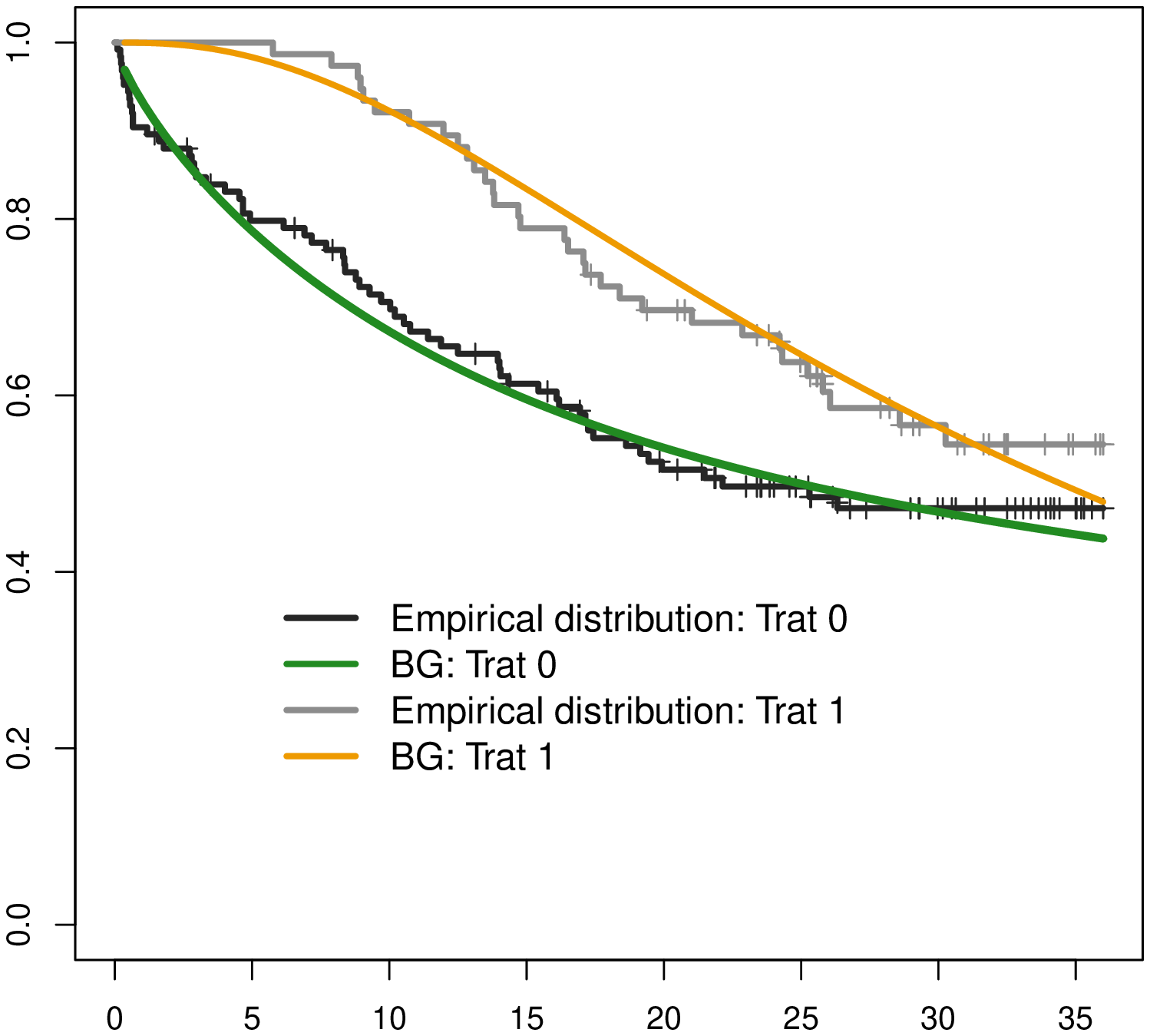}~
			\caption{Estimated survival function for the $\text{BGamma}(\tetn_\delta)$ regression model and the empirical survival. (a) Adjustment considering only the $\alpha$ parameter. (b) Adjustment considering both $\alpha$ and $\beta$ parameters  for gastric cancer data.\label{dens_aplica_temp}}
		\end{center}
	\end{figure}
	
%

\section{Concluding remarks}\label{sec:08}
In this work, we have introduced a bimodal generalization of the gamma distribution that can be an alternative to model bimodal data. It was obtained using a quadratic transformation based on the alpha-skew-normal model. Since this generalization has three parameters, the parameter estimation is simpler than in mixtures. We have discussed the properties of this density such as bimodality, moment generating function, hazard rate and entropy measures. In order to check the efficiency of the maximum likelihood estimators, we have carried out a Monte Carlo simulation study. We have also introduced a regression model based on the proposed bimodal gamma distribution. The fitting of the distribution along with its regression model was tested with two real data sets and it was shown that our model may outperform some distributions found in literature. Thus, we have a flexible distribution that presented consistent results in data modeling.


\section*{Acknowledgements}
\noindent
This study was financed in part by the Coordena\c{c}\~{a}o de Aperfei\c{c}oamento de Pessoal de N\'{i}vel Superior
- Brasil (CAPES) - Finance Code 001.

\bibliographystyle{apalike}

\end{document}